\RequirePackage{fix-cm}
\documentclass[twocolumn]{svjour3}          %
\usepackage{amssymb}
\usepackage{natbib}
\usepackage[dvipsnames]{xcolor}
\usepackage{amsmath}
\usepackage{mathtools}
\usepackage{thmtools}
\usepackage{commath}
\usepackage[caption=false,font=footnotesize]{subfig}
\usepackage{verbatim}
\usepackage[dvipsnames]{xcolor}

\newcommand{\E}{\mathbb{E}}
\newcommand{\C}{\mathbb{C}}
\newcommand{\V}{\mathbb{V}}
\newcommand{\T}{\ensuremath{\mathsf{T}}}

\smartqed  %
\usepackage{graphicx}
\newif\iffinal %

\usepackage{color}

\usepackage{tikz}
\usepackage{pgfplots}
\usepgfplotslibrary{external}
\usepgfplotslibrary{statistics}
\usetikzlibrary{decorations}
\usetikzlibrary{decorations.markings}
\usetikzlibrary{pgfplots.statistics}
\pgfplotsset{compat = 1.3}

\definecolor{AaltoYellow}{RGB}{255, 205, 0}
\definecolor{AaltoRed}{RGB}{239, 51, 64}
\definecolor{AaltoBlue}{RGB}{0, 94, 184}
\definecolor{AaltoPurple}{RGB}{117, 59, 189}
\definecolor{AaltoGreen}{RGB}{120, 190, 32}
\definecolor{AaltoPurple2}{RGB}{192, 26, 162}

\begin{document}

\title{Probabilistic Solutions To Ordinary Differential Equations\\ As Non-Linear Bayesian Filtering: A New Perspective%
}

\author{Filip Tronarp         \and
        Hans Kersting \and
        Simo S\"arkk\"a \and
	Philipp Hennig%
}

\institute{Filip Tronarp   \at Aalto University \\
              02150 Espoo, Finland  \\
              \email{filip.tronarp@aalto.fi}           %
           \and
           Hans Kersting \at University of T\"ubingen \\
           and Max Planck Institute for Intelligent Systems \\
           72076 T\"ubingen, Germany\\
	  \email{hkersting@tuebingen.mpg.de}
	  \and
   	  Simo S\"arkk\"a \at Aalto University \\
      02150 Espoo, Finland  \\
              \email{simo.sarkka@aalto.fi}
          \and
          Philipp Hennig \at University of T\"ubingen \\
           and Max Planck Institute for Intelligent Systems \\
           72076 T\"ubingen, Germany\\
  	\email{ph@tue.mpg.de}
}

\date{Received: date / Accepted: date}

\maketitle

\begin{abstract}
We formulate probabilistic numerical approximations to solutions of ordinary differential equations (ODEs) as problems in Gaussian process (GP) regression with non-linear measurement functions. This is achieved by defining the measurement sequence to consist of the observations of the difference between the derivative of the GP and the vector field evaluated at the GP---which are all identically zero at the solution of the ODE. When the GP has a state-space representation, the problem can be reduced to a non-linear Bayesian filtering problem and all widely-used approximations to the Bayesian filtering and smoothing problems become applicable. Furthermore, all previous GP-based ODE solvers that are formulated in terms of generating synthetic measurements of the gradient field come out as specific approximations. Based on the non-linear Bayesian filtering problem posed in this paper, we develop novel Gaussian solvers for which we establish favourable stability properties. Additionally, non-Gaussian approximations to the filtering problem are derived by the particle filter approach. The resulting solvers are compared with other probabilistic solvers in illustrative experiments.
\end{abstract}

\section{Introduction}
\label{sec:1}
We consider an initial value problem (IVP), that is, an ordinary differential equation (ODE)
\begin{equation}\label{eq:ode}
  \dot{y}(t) = f\left(y(t),t\right),\ \forall t \in [0,T], \quad y(0) = y_0 \in \mathbb R^d,
\end{equation}
with initial value $y_0$ and vector field $f:\mathbb R^d \times \mathbb R_+ \to \mathbb R^d$.
Numerical solvers for IVPs approximate $y:[0,T] \to \mathbb R^d$ and are of paramount importance in almost all areas of science and engineering.
Extensive knowledge about this topic has been accumulated in numerical analysis literature, for example, in \citet{hairer87:_solvin_ordin_differ_equat_i}, \citet{DeuflhardBornemann2002}, and \citet{Butcher2008}.
However, until recently, a probabilistic quantification of the inevitable uncertainty--for all but the most trivial ODEs--from the numerical error over their outputs has been omitted.

Moreover, ODEs are often part of a pipeline surrounded by preceding and subsequent computations, which are themselves corrupted by uncertainty from model misspecification, measurement noise, approximate inference or, again, numerical inaccuracy \citep{KennedyOHagan02}.
In particular, ODEs are often integrated using estimates of its parameters rather than the correct ones.
See \citet{ZhangJadbabaie18} and \citet{ChenDuvenaud18} for recent examples of such computational chains involving ODEs.
The field of \emph{probabilistic numerics} (PN) \citep{HenOsbGirRSPA2015} seeks to overcome this ignorance of numerical uncertainty and the resulting overconfidence by providing \emph{probabilistic numerical methods}. These solvers quantify numerical errors probabilistically and add them to uncertainty from other sources.
Thereby, they can take decisions in a more uncertainty-aware and uncertainty-robust manner \citep{paul_alternating_2018}.

In the case of ODEs, one family of probabilistic solvers (\citealt{skilling1991bayesian}, \citealt{HennigAISTATS2014}, and \citealt{schober2014nips}) first treated IVPs as Gaussian process (GP) regression \cite[Chapter 2]{RasmussenWilliams}.
Then, \citet{Kersting2016UAI} and \citet{schober2018} sped up these methods by regarding them as stochastic filtering problems \citep{oksendal2003stochastic}.
These completely deterministic filtering methods converge to the true solution with high polynomial rates \citep{Kersting2018}. In their methods data for the 'Bayesian update' is constructed by evaluating the vector field $f$ under the GP predictive mean of $y(t)$ and linked to the model with a Gaussian likelihood \cite[Section 2.3]{schober2018}. See also \citet[Section 1.2]{Wang18} for alternative likelihood models. 
This conception of data implies that it is the output of the adopted inference procedure. More specifically, one can show that with everything else being equal, two different priors may end up operating on different measurement sequences. Such a coupling between prior and measurements is not standard in statistical problem formulations, as acknowledged in \citet[Section 2.2]{schober2018}. It makes the model and the subsequent inference difficult to interpret. For example, it is not clear how to do Bayesian model comparisons \citep[Section 2.4]{Cockayne2017BayesianPN} when two different priors necessarily operate on two different data sets for the same inference task. 

Instead of formulating the solution of Eq. \eqref{eq:ode} as a Bayesian GP regression problem, another line of work on probabilistic solvers for ODEs comprising the methods from \citealt{o.13:_bayes_uncer_quant_differ_equat}, \citealt{conrad_probability_2017}, \citealt{teymur2016probabilistic}, \citealt{Lie17}, \citealt{AbdulleGaregnani17}, and \citealt{Teymur2018a} aims to represent the uncertainty arising from the discretization error by a set of samples.
While multiplying the computational cost of classical solvers with the amount of samples, these methods can capture arbitrary (non-Gaussian) distributions over the solutions and can reduce over-confidence in inverse problems for ODEs---as demonstrated in \citet[Section 3.2.]{conrad_probability_2017}, \citet[Section 7]{AbdulleGaregnani17}, and \citet{Teymur2018a}.
These solvers can be considered as more expensive, but statistically more expressive.
This paper contributes a particle filter as a sampling-based filtering method at the intersection of both lines of work, providing a previously missing link.

The contributions of this paper are the following:
Firstly, we circumvent the issue of generating synthetic data, by recasting solutions of ODEs in terms of non-linear Bayesian filtering problems in a well defined state-space model. For any fixed-time discretisation, the measurement sequence and likelihood are also fixed. That is, we avoid the coupling of prior and measurement sequence, that is for example present in \citet{schober2018}.
This enables application of all Bayesian filtering and smoothing techniques to ODEs as described, for example, in \citet{Sarkka2013}.
Secondly, we show how the application of certain inference techniques recovers the previous filtering-based methods.
Thirdly, we discuss novel algorithms giving rise to both Gaussian and non-Gaussian solvers.

Fourthly, we establish a stability result for the novel Gaussian solvers.
Fifthly, we discuss practical methods for uncertainty calibration, and in the case of Gaussian solvers, we give explicit expressions. 
Finally, we present some illustrative experiments demonstrating that these methods are practically useful both for fast inference of the unique solution of an ODE as well as for representing multi-modal distributions of trajectories.

\section{Bayesian Inference for Initial Value Problems}\label{sec:bayesian_ivps}
Formulating an approximation of the solution to Eq. \eqref{eq:ode} at a discrete set of points $\{t_n\}_{n=0}^N$ as a problem of Bayesian inference requires, as always, three things: a prior measure, data, and a likelihood, which define a posterior measure through Bayes' rule.

We start with examining a continuous-time formulation in Section \ref{subsec:CT}, where Bayesian conditioning should, in the ideal case, give a Dirac measure at the true solution of Eq. \eqref{eq:ode} as the posterior. This has two issues: (1) conditioning on the entire gradient field is not feasible on a computer in finite time and (2) the conditioning operation itself is intractable. Issue (1) is present in classical Bayesian quadrature \citep{Briol2015probint} as well. Limited computational resources imply that only a finite number of evaluations of the integrand can be used. Issue (2) turns, what is linear GP regression in Bayesian quadrature, into non-linear GP regression. While this is unfortunate, it appears reasonable that something should be lost as the inference problem is more complex.

With this in mind, a discrete-time non-linear Bayesian filtering problem is posed in Section \ref{subsec:DT}, which targets the solution of Eq. \eqref{eq:ode} at a discrete set of points.

\subsection{A Continuous-Time Model}\label{subsec:CT}
Like previous works mentioned in Section \ref{sec:1}, we consider priors given by a GP
\begin{equation*}
X(t) \sim \mathrm{GP}\left(\bar{x},k\right),
\end{equation*}
where $\bar{x}(t)$ is the mean function and $k(t,t')$ is the covariance function. The vector $X(t)$ is given by
\begin{equation}
X(t) = \begin{bmatrix}\big(X^{(1)}(t)\big)^\T, \dots,\big(X^{(q+1)}(t)\big)^\T\end{bmatrix}^\T,
\end{equation}
where $X^{(1)}(t)$ and $X^{(2)}(t)$ model $y(t)$ and $\dot{y}(t)$, respectively. The remaining $q-1$ sub-vectors in $X(t)$ can be used to model higher order derivatives of $y(t)$ as done by \citet{schober2018} and \citet{Kersting2016UAI}. We define such priors by a stochastic differential equation \citep{oksendal2003stochastic}, that is,
\begin{subequations}\label{eq:ode_prior}
\begin{align}
X(0) &\sim \mathcal{N}\big(\mu^-(0),\Sigma^-(0) \big), \\
\dif X(t) &= \big[F X(t) + u\big] \dif t + L \dif B(t),
\end{align}
\end{subequations}
where $F$ is a state transition matrix, $u$ is a forcing term, $L$ is a diffusion matrix, and $B(t)$ is a vector of standard Wiener processes.

Note that for $X^{(2)}(t)$ to be the derivative of $X^{(1)}$, $F$, $u$, and $L$ are such that
\begin{equation}\label{eq:model_assumption1}
\dif X^{(1)}(t) = X^{(2)}(t) \dif t.
\end{equation}
The use of an SDE---instead of a generic GP prior---is computationally advantageous because it restricts the priors to Markov processes due to \citet[Theorem 7.1.2]{oksendal2003stochastic}.
This allows for inference with linear time-complexity in $N$, while the time-complexity is $N^3$ for GP priors in general \citep{hartikainen2010kalman}.

Inference requires data, and an associated likelihood. Previous authors, such as \citet{schober2018} and \citet{o.13:_bayes_uncer_quant_differ_equat}, put forth the view of the prior measure defining an \emph{inference agent}, which cycles through extrapolating, generating measurements of the vector field, and updating. Here we argue that there is no need for generating measurements, since re-writing Eq. \eqref{eq:ode} yields the requirement
\begin{equation}
\dot{y}(t) - f(y(t),t) = 0.
\end{equation}
This suggests that a measurement relating the prior defined by Eq. \eqref{eq:ode_prior} to the solution of Eq. \eqref{eq:ode} ought to be defined as
\begin{equation}\label{eq:ct_measurement}
Z(t) = X^{(2)}(t) - f(X^{(1)}(t),t).
\end{equation}
While conditioning the process $X(t)$ on the event $Z(t) = 0$ for all $t \in [0,T]$ can be
formalised using the concept of \emph{disintegration} \citep{Cockayne2017BayesianPN}, 
it is intractable in general and thus impractical for computer implementation. Therefore, we formulate a discrete-time
inference problem in the sequel.
\subsection{A Discrete-Time Model}\label{subsec:DT}
In order to make the inference problem tractable, we only attempt to condition the process $X(t)$ on $Z(t) = z(t) \triangleq 0$ at a set of discrete time-points, $\{t_n\}_{n=0}^N$.
We consider a uniform grid, $t_{n+1} = t_n + h$, though extending the present methods to non-uniform grids can be done as described in \citet{schober2018}. In the sequel, we will denote a function evaluated at $t_n$ by subscript $n$, for example $z_n = z(t_n)$. From Eq. \eqref{eq:ode_prior} an equivalent discrete-time system can be obtained \citep[Chapter 3.7.3]{grewal2001kalman} \footnote{Here `equivalent' is used in the sense that the probability distribution of the continuous-time process evaluated on the grid coincides with the probability distribution of the discrete-time process \citep[Page 17]{sarkka2006thesis}.}. The inference problem becomes
\begin{subequations}\label{eq:inference_problem}
\begin{align}
X_0 &\sim \mathcal{N}(\mu^F_0,\Sigma^F_0), \\
X_{n+1} \mid X_n &\sim \mathcal{N}\big(A(h)X_n+\xi(h), Q(h)\big), \\
Z_n \mid X_n &\sim \mathcal{N}\big(\dot{C}X_n - f(C X_n,t_n),R\big), \label{eq:dt_measurement}\\
z_n &\triangleq 0,\quad n = 1,\dots,N, \label{eq:z=0}
\end{align}
\end{subequations}
where $z_n$ is the realisation of $Z_n$. The parameters $A(h)$, $\xi(h)$, and $Q(h)$ are given by
\begin{subequations}\label{eq:lti_disc}
\begin{align}
A(h)   &= \exp (Fh), \\
\xi(h) &= \int_0^h \exp(F(h-\tau)) u \dif \tau,\\
Q(h)   &= \int_0^h \exp(F(h-\tau)) L L^\T \exp(F^\T(h-\tau)) \dif \tau.
\end{align}
\end{subequations}
Furthermore, $C = [\mathrm{I}  \ 0 \ \ldots\ 0]$ and $\dot{C} = [ 0 \ \mathrm{I}  \ 0 \ \ldots\ 0]$. That is, $CX_n = X_n^{(1)}$ and $\dot{C}X_n =  X_n^{(2)}$. 
A measurement variance, $R$, has been added to $Z(t_n)$ for greater generality, which simplifies the construction of particle filter algorithms.
The likelihood model in Eq. \eqref{eq:dt_measurement} has previously been used in the gradient matching approach
to inverse problems to avoid explicit numerical integration of the ODE  (see, e.g.,~\citealt{calderhead2008accelerating}).

The inference problem posed in Eq. \eqref{eq:inference_problem} is a standard problem in non-linear GP regression
\citep{RasmussenWilliams}, also known as Bayesian filtering and smoothing in stochastic signal processing \citep{Sarkka2013}.
Furthermore, it reduces to Bayesian quadrature when the vector field does not depend on $y$. This is Proposition \ref{prop:BQ} below.

\begin{proposition}\label{prop:BQ}
Let $X^{(1)}_0 = 0$, $f(y(t),t) = g(t)$, $y(0) = 0$, and $R = 0$. Then the posteriors of $\{X^{(1)}_n\}_{n=1}^N$ are Bayesian quadrature approximations for
\begin{equation}\label{eq:BQ}
\int_0^{nh} g(\tau) \dif \tau,\quad n=1,\dots,N.
\end{equation}
\end{proposition}

A proof of Proposition \ref{prop:BQ} is given in Appendix \ref{appendix:proof1}.

\begin{remark}
The Bayesian quadrature method described in Proposition \ref{prop:BQ} conditions on function evaluations outside the domain of integration for $n < N$. This corresponds to the smoothing equations associated with Eq. \eqref{eq:inference_problem}. If the integral on the domain $[0,nh]$ is only conditioned on evaluations of $g$ inside the domain then the filtering estimates associated with Eq. \eqref{eq:inference_problem} are obtained.
\end{remark}

\subsection{Gaussian Filtering}\label{subsec:gaussian_filtering}
The inference problem posed in Eq. \eqref{eq:inference_problem} is a standard problem in statistical signal processing and
machine learning, and the solution is often approximated by Gaussian filters and smoothers \citep{Sarkka2013}.
Let us define $z_{1:n} = \{ z_l \}_{l=1}^n$ and the following conditional moments
\begin{subequations}
\begin{align}
\mu_n^F &\triangleq \E[X_n \mid z_{1:n}], \\
\Sigma_n^F &\triangleq \V[X_n \mid z_{1:n}], \\
\mu_n^P &\triangleq \E[X_n \mid z_{1:n-1}], \\
\Sigma_n^P &\triangleq \V[X_n \mid z_{1:n-1}],
\end{align}
\end{subequations}
where $\E[\cdotp \mid  z_{1:n}]$ and $\V[ \cdotp \mid  z_{1:n}]$ are the conditional mean and covariance operators given the measurements
$Z_{1:n} = z_{1:n}$. Additionally, $\E[ \cdotp \mid z_{1:0}] =  \E[ \cdotp ]$ and $\V[ \cdotp \mid z_{1:0}] =  \V[ \cdotp ]$ by convention. Furthermore, $\mu_n^F$ and $\Sigma_n^F$ are referred to as the filtering mean and covariance, respectively.
Similarly,  $\mu_n^P$ and $\Sigma_n^P$ are referred to as the predictive mean and covariance, respectively.
In Gaussian filtering, the following relationships hold between $\mu_n^F$ and $\Sigma_n^F$, and $\mu_{n+1}^P$ and $\Sigma_{n+1}^P$:
\begin{subequations}\label{eq:kalman_prediction}
\begin{align}
\mu_{n+1}^P &= A(h) \mu_n^F + \xi(h), \\
\Sigma_{n+1}^P &= A(h) \Sigma_n^F A^\T(h) + Q(h),
\end{align}
\end{subequations}
which are the prediction equations \citep[Eq. 6.6]{Sarkka2013}. The update equations, relating the predictive moments
$\mu_{n}^P$ and $\Sigma_{n}^P$ with the filter estimate, $\mu_{n}^F$, and its covariance $\Sigma_{n}^F$, are given by \citep[Eq. 6.7]{Sarkka2013}
\begin{subequations}\label{eq:general_gaussian_update}
\begin{align}
S_n &= \V\big[\dot{C}X_n - f(CX_n,t_n)\mid z_{1:n-1}\big] + R, \\
K_n &= \C\big[X_n,\dot{C}X_n - f(CX_n,t_n)\mid z_{1:n-1}\big] S_n^{-1}, \\
\hat{z}_n &= \E\big[\dot{C}X_n - f(CX_n,t_n)\mid z_{1:n-1}\big], \\
\mu_n^F &= \mu_n^P + K_n(z_n - \hat{z}_n), \\
\Sigma_n^F &= \Sigma_n^P - K_n S_n K_n^\T,
\end{align}
\end{subequations}
where the expectation ($\E$), covariance ($\V$) and cross-covariance ($\C$) operators are with respect to $X_n \sim \mathcal{N}(\mu_n^P,\Sigma_n^P)$. Evaluating these moments is intractable in general, though various approximation schemes exist in literature. Some standard approximation methods shall be examined below. In particular, the methods of \citet{schober2018} and \citet{Kersting2016UAI} come out as particular approximations to Eq. \eqref{eq:general_gaussian_update}.

\subsection{Taylor-Series Methods}
\label{sec:Taylor-Series_Methods}
A classical method in filtering literature to deal with non-linear measurements of the form in Eq. \eqref{eq:inference_problem} is to make a first order Taylor-series expansion, thus turning the problem into a standard update in linear filtering. However, before going through the details of this it is instructive to interpret the method of \citet{schober2018} as an even simpler Taylor-series method. This is Proposition \ref{prop:schober} below.
\begin{proposition}\label{prop:schober}
Let $R = 0$ and approximate $f(C X_n,t_n)$ by its zeroth order Taylor expansion in $X_n$ around the point $\mu_n^P$
\begin{equation}
f\big(CX_n,t_n\big) \approx f\big(C\mu_n^P,t_n\big).
\end{equation}
Then, the approximate posterior moments are given by
\begin{subequations}\label{eq:schober}
\begin{align}
S_n &\approx \dot{C} \Sigma_n^P \dot{C}^\T + R, \\
K_n &\approx \Sigma_n^P \dot{C}^\T S_n^{-1}, \\
\hat{z}_n &\approx \dot{C} \mu_n^P - f\big(C\mu_n^P,t_n\big), \\
\mu_n^F &\approx \mu_n^P + K_n(z_n - \hat{z}_n), \\
\Sigma_n^F &\approx \Sigma_n^P - K_n S_n K_n^\T,
\end{align}
\end{subequations}
which is precisely the update by \citet{schober2018}.
\end{proposition}

\paragraph{A First Order Approximation.} The approximation in Eq. \eqref{eq:schober} can be refined by using a first order approximation,
which is known as the extended Kalman filter (EKF) in signal processing literature \citep[Algorithm 5.4]{Sarkka2013}. That is,
\begin{equation}
\begin{split}
f\big(CX_n,t_n\big) &\approx f\big(C\mu_n^P,t_n\big) \\
&\quad+ J_f\big(C\mu_n^P,t_n\big)C\big(X_n - \mu_n^P\big),
\end{split}
\end{equation}
where $J_f$ is the Jacobian of $y \to f(y,t)$. The filter update is then
\begin{subequations}\label{eq:ekf}
\begin{align}
\tilde{C}_n &= \dot{C} -   J_f\big(C\mu_n^P,t_n\big)C, \\
S_n &\approx \tilde{C}_n \Sigma_n^P\tilde{C}_n^\T + R, \\
K_n &\approx \Sigma_n^P \tilde{C}_n^\T S_n^{-1}, \\
\hat{z}_n &\approx \dot{C} \mu_n^P - f\big(C\mu_n^P,t_n\big), \\
\mu_n^F &\approx \mu_n^P + K_n(z_n - \hat{z}_n), \\
\Sigma_n^F &\approx \Sigma_n^P - K_n S_n K_n^\T.
\end{align}
\end{subequations}
Hence the extended Kalman filter computes the residual, $z_n - \hat{z}_n$, in the same manner as \citet{schober2018}.
However, as the filter gain, $K_n$, now depends on evaluations of the Jacobian, the resulting probabilistic ODE solver is different in general.

While Jacobians of the vector field are seldom exploited in ODE solvers, they play a central role in Rosenbrock methods,
(\citealt{Rosenbrock1963} and \citealt{Hochbruck2009}). The Jacobian of the vector field was also recently used by \citet{Teymur2018a}
for developing a probabilistic solver.

Although the extended Kalman filter goes as far back as the 1960s \citep{jazwinski1970stochastic}, the update in Eq.
\eqref{eq:ekf} results in a probabilistic method for estimating the solution of \eqref{eq:ode} that appears to be novel. Indeed, to the best of the authors' knowledge, the only Gaussian filtering based solvers that have appeared so far are those by  \citealt{Kersting2016UAI}, \citealt{magnani2017}, and \citet{schober2018}.

\subsection{Numerical Quadrature}\label{subsec:quadrature}
Another method to approximate the quantities in Eq. \eqref{eq:general_gaussian_update} is by quadrature,
which consists of a set of nodes $\{\mathcal{X}_{n,j}\}_{j=1}^J$ with weights $\{w_{n,j}\}_{j=1}^J$ that are associated
to the distribution $\mathcal{N}(\mu_n^P,\Sigma_n^P)$. These nodes and weights can either be constructed to integrate polynomials
up to some order exactly (see, e.g., \citealt{McNameeStenger1967} and \citealt{Golub1969}), or by
Bayesian quadrature \citep{Briol2015probint}. In either case, the expectation of a function $\psi(X_n)$ is approximated by
\begin{equation}
\E[\psi(X_n)] \approx \sum_{j=1}^J w_{n,j} \psi(\mathcal{X}_{n,j}).
\end{equation}
Therefore, by appropriate choices of $\psi$ the quantities in Eq. \eqref{eq:general_gaussian_update} can be approximated. We shall refer to filters using a third degree fully symmetric rule \citep{McNameeStenger1967} as Unscented Kalman filters (UKF), which is the name that was adopted when it was first introduced to the signal processing community \citep{Julier2000}.
For a suitable cross-covariance assumption and a particular choice of quadrature, the method of \citet{Kersting2016UAI} is retrieved. This is Proposition \ref{prop:kersting}.

\begin{proposition}\label{prop:kersting}
Let $\{\mathcal{X}_{n,j}\}_{j=1}^J$ and $\{w_{n,j}\}_{j=1}^J$ be the nodes and weights, corresponding to a Bayesian quadrature rule
with respect to $\mathcal{N}(\mu_n^P,\Sigma_n^P)$. Furthermore, assume $R = 0$ and that the cross-covariance between $\dot{C}X_n$ and $f(C X_n,t_n)$ is approximated as zero,
\begin{equation}
\C\big[\dot{C}X_n,f(CX_n,t_n)\mid z_{1:n-1}\big] \approx 0.
\end{equation}
Then the probabilistic solver proposed in \citet{Kersting2016UAI} is a Bayesian quadrature approximation to Eq. \eqref{eq:general_gaussian_update}.
\end{proposition}

A proof of Proposition \ref{prop:kersting} is given in Appendix \ref{appendix:proof3}.

While a cross-covariance assumption of Proposition \ref{prop:kersting} reproduces the method of \citet{Kersting2016UAI},
Bayesian quadrature approximations have previously been used for Gaussian filtering in signal processing applications  by \citet{PS15},
which in this context gives a new solver.

\subsection{Affine Vector Fields}\label{subsec:affine_vector_fields}
It is instructive to examine the particular case when the vector field in Eq. \eqref{eq:ode} is affine. That is,
\begin{equation}\label{eq:affine_vector_field}
f(y(t),t) = \Lambda(t) y(t) + \zeta(t).
\end{equation}
In such a case, Eq. \eqref{eq:inference_problem} becomes a linear Gaussian system, which is solved exactly by a Kalman filter. The equations for implementing this Kalman filter are precisely Eq. \eqref{eq:kalman_prediction} and Eq. \eqref{eq:general_gaussian_update}, although the latter set of equations can be simplified. Define  $H_n = \dot{C} - \Lambda(t_n) C$, then the update equations become
\begin{subequations}\label{eq:affine_gaussian_update}
\begin{align}
S_n &= H_n \Sigma_n^P H_n^\T + R, \\
K_n &= \Sigma_n^P H_n^\T S_n^{-1}, \\
\mu_n^F &= \mu_n^P + K_n \big(\zeta(t_n) - H_n \mu_n^P  \big)\\
\Sigma_n^F &= \Sigma_n^P - K_n S_n K_n^\T.
\end{align}
\end{subequations}
\begin{lemma}\label{lem:ekfukf_affine}
Consider the inference problem in Eq. \eqref{eq:inference_problem} with an affine vector field as given in Eq. \eqref{eq:affine_vector_field}. %
Then the EKF reduces to the exact Kalman filter, which uses the update in Eq. \eqref{eq:affine_gaussian_update}. Furthermore, the same holds for Gaussian filters using a quadrature approximation to Eq. \eqref{eq:general_gaussian_update}, provided that it integrates polynomials correctly up to second order with respect to the distribution $\mathcal{N}(\mu_n^P,\Sigma_n^P)$. 
\end{lemma}
\begin{proof}
Since the Kalman filter, the EKF, and the quadrature approach all use Eq. \eqref{eq:kalman_prediction} for prediction, it is sufficient to make sure that the EKF and the quadrature approximation compute Eq. \eqref{eq:general_gaussian_update} exactly, just as the Kalman filter. Now the EKF approximates the vector field by an affine function for which it computes the moments in Eq. \eqref{eq:general_gaussian_update} exactly. Since this affine approximation is formed by a truncated Taylor series, it is exact for affine functions and the statement pertaining to the EKF holds. Furthermore, the Gaussian integrals in Eq. \eqref{eq:general_gaussian_update} are polynomials of degree at most two for affine vector fields and are therefore computed exactly by the quadrature rule by assumption. \qed
\end{proof}
\subsection{Particle Filtering}
The Gaussian filtering methods from Section \ref{subsec:gaussian_filtering} may often suffice. However, there are cases where more sophisticated inference methods may be preferable, for instance, when the posterior becomes multi-modal due to chaotic behavior or `numerical bifurcations'. That is, when it is numerically unknown whether the true solution is above or below a certain threshold that determines the limit behaviour of its trajectory. While sampling based probabilistic solvers such as those of 
\citet{o.13:_bayes_uncer_quant_differ_equat}, \citet{conrad_probability_2017}, \citet{teymur2016probabilistic}, \citet{Lie17}, \citet{AbdulleGaregnani17}, and \citet{Teymur2018a} can pick up such phenomena, the Gaussian filtering based ODE solvers discussed in Section \ref{subsec:gaussian_filtering} cannot. However, this limitation may be overcome by approximating the filtering distribution of the inference problem in Eq. \eqref{eq:inference_problem} with particle filters
that are based on a sequential formulation of importance sampling \citep{Doucet2001}.

A particle filter operates on a set of particles, $\{X_{n,j}\}_{j=1}^J$, a set of positive weights $\{w_{n,j}\}_{j=1}^J$ associated to the particles that sum to one and an importance density, $g(x_{n+1}\mid x_n, z_n)$. The particle filter then cycles through three steps (1) propagation, (2) re-weighting, and (3) re-sampling \citep[Chapter 7.4]{Sarkka2013}.

The propagation step involves sampling particles at time $n+1$ from the importance density:
\begin{equation}
X_{n+1,j} \sim g(x_{n+1}\mid X_{n,j}, z_n).
\end{equation}
The re-weighting of the particles is done by a likelihood ratio with the product of the measurement density and the transition density of Eq. \eqref{eq:inference_problem}, and the importance density. That is, the updated weights are given by
\begin{subequations}\label{eq:weight_update}
\begin{align}
\rho(x_{n+1},x_n) &= \frac{p(z_{n+1}\mid x_{n+1}) p(x_{n+1}\mid x_{n}) }{g(x_{n+1}\mid x_{n},z_{n+1})},\label{eq:pflr}\\
w_{n+1,j} &\propto \rho\big(X_{n+1,j},X_{n,j}\big) w_{n,j},
\end{align}
\end{subequations}
where the proportionality sign indicates that the weights need to be normalised to sum to one after they have been updated according to Eq. \eqref{eq:weight_update}.  
The weight update is then followed by an optional re-sampling step \citep[Chapter 7.4]{Sarkka2013}. While not re-sampling in principle yields a valid algorithm, it becomes necessary in order to avoid the degeneracy problem for long time series \citep[Chapter 1.3]{Doucet2001}.   
The efficiency of particle filters depends on the choice of importance density. In terms of variance, the locally optimal importance density is given by \citep{Doucet2001}
\begin{equation}\label{eq:optimal_id}
g(x_n\mid x_{n-1},z_{n}) \propto p(z_n\mid x_n)p(x_n\mid x_{n-1}).
\end{equation}
While Eq. \eqref{eq:optimal_id} is almost as intractable as the full filtering distribution, the Gaussian filtering methods from Section \ref{subsec:gaussian_filtering} can be used to make a good approximation. For instance, the approximation to the optimal importance density using Eq. \eqref{eq:schober} is given by
\begin{subequations}\label{eq:schober_proposal}
\begin{align}
S_n &= \dot{C} Q(h)\dot{C}^\T + R, \\
K_n &= Q(h) \dot{C}^\T S_n^{-1}, \\
\hat{z}_n &= \dot{C} A(h)x_{n-1} - f\big(CA(h)x_{n-1},t_n\big), \\
\mu_n &= A(h)x_{n-1} + K_n(z_n - \hat{z}_n), \\
\Sigma_n &= Q(h) - K_n S_n K_n^\T, \\
&g(x_n\mid x_{n-1},z_n) = \mathcal{N}(x_n;\mu_n,\Sigma_n).
\end{align}
\end{subequations}
An importance density can be similarly constructed from Eq. \eqref{eq:ekf}, resulting in:
\begin{subequations}\label{eq:ekf_proposal}
\begin{align}
\tilde{C}_n &= \dot{C} -   J_f\big(CA(h)x_{n-1} ,t_n\big)C, \\
S_n &= \tilde{C}_n Q(h) \tilde{C}_n^\T + R, \\
K_n &= Q(h) \tilde{C}_n^\T S_n^{-1}, \\
\hat{z}_n &= \dot{C} A(h)x_{n-1}  - f\big(CA(h)x_{n-1} ,t_n\big), \\
\mu_n &= A(h)x_{n-1} + K_n(z_n - \hat{z}_n), \\
\Sigma_n &= Q(h) - K_n S_n K_n^\T,\\
&g(x_n\mid x_{n-1},z_n) = \mathcal{N}(x_n;\mu_n,\Sigma_n).
\end{align}
\end{subequations}
Note that we have assumed $\xi(h) = 0$ in Eqs. \eqref{eq:schober_proposal} and \eqref{eq:ekf_proposal}, which can be extended to $\xi(h) \neq 0$ by replacing $A(h)x_{n-1}$ with $A(h)x_{n-1} + \xi(h)$. We refer the reader to \citet[Section II.D.2]{Doucet2000} for a more thorough discussion on the use of local linearisation methods to construct importance densities.

We conclude this section with a brief discussion on the convergence of particle filters. The following theorem is given by \citet{Crisan2002}. 
\begin{theorem}\label{thm:m2convergence}
Let $\rho(x_{n+1},x_n)$ in Eq. \eqref{eq:pflr} be bounded from above and denote the true filtering measure associated with Eq. \eqref{eq:inference_problem} at time $n$ by $p_n^R$ and let $\hat{p}_n^{R,J}$ be its particle approximation using $J$ particles with importance density $g(x_{n+1}\mid x_{n},z_{n+1})$. Then, for all $n \in \mathbb N_0$, there exists a constant $c_n$ independent of $J$ such that for any bounded Borel function $\phi \colon \mathbb{R}^{d(q+1)} \to \mathbb{R}$ the following bound holds
\begin{equation}
\E_{\mathrm{MC}}[(\langle \hat{p}_n^{R,J},\phi \rangle - \langle p_n^R,\phi \rangle)^2]^{1/2} \leq c_n J^{-1/2}\norm{\phi},
\end{equation}
where $\langle p, \phi\rangle$ denotes $\phi$ integrated with respect to $p$ and $\E_{\mathrm{MC}}$ denotes the expectation over realisations of the particle method, and $\norm{\cdot}$ is the supremum norm.
\end{theorem}
Theorem \ref{thm:m2convergence} shows that we can decrease the distance (in the weak sense) between $\hat{p}_n^{R,J}$ and $p_n^R$ by increasing $J$. However, the object we want to approximate is $p_n^0$ (the exact filtering measure associated with Eq. \eqref{eq:inference_problem} for $R=0$) but setting $R = 0$ makes the likelihood ratio in Eq. \eqref{eq:pflr} ill-defined for the proposal distributions in Eqs. \eqref{eq:schober_proposal} and \eqref{eq:ekf_proposal}. 
This is because, when $R=0$, then $p(z_{n+1}\mid x_{n+1})p(x_{n+1}\mid x_n)$ has its support on the surface $\dot{C}x_{n+1} = f(Cx_{n+1},t_{n+1})$ while Eqs. \eqref{eq:schober_proposal} or \eqref{eq:ekf_proposal} imply that the variance of $\dot{C}X_{n+1}$ or $\tilde{C}_{n+1}X_{n+1}$ will be zero with respect to $g(x_{n+1}\mid x_n, z_{n+1})$, respectively. That is, $g(x_{n+1}\mid x_n, z_{n+1})$ is supported on a hyperplane. It follows that the null-sets of $g(x_{n+1}\mid x_n, z_{n+1})$ are not necessarily null-sets of $p(z_{n+1}\mid x_{n+1})p(x_{n+1}\mid x_n)$ and the likelihood ratio in Eq. \eqref{eq:pflr} can therefore be undefined.
However, a straightforward application of the triangle inequality together with Theorem \ref{thm:m2convergence} gives
\begin{equation}
\begin{split}
&\E_{\mathrm{MC}}[(\langle \hat{p}_n^{R,J},\phi \rangle - \langle p_n^0,\phi \rangle)^2]^{1/2} \\
\quad&\leq \E_{\mathrm{MC}}[(\langle \hat{p}_n^{R,J},\phi \rangle - \langle p_n^R,\phi \rangle)^2]^{1/2} \\
\qquad&+ \E_{\mathrm{MC}}[(\langle p_n^{R},\phi \rangle - \langle p_n^0,\phi \rangle)^2]^{1/2} \\
\ &= \E_{\mathrm{MC}}[(\langle \hat{p}_n^{R,J},\phi \rangle - \langle p_n^R,\phi \rangle)^2]^{1/2} \\
\qquad &+  \big\lvert\langle p_n^{R},\phi \rangle - \langle p_n^0,\phi \rangle\big\rvert\\
&\leq c_n J^{-1/2}\norm{\phi} + \big\lvert\langle p_n^R,\phi \rangle -\langle p_n^0,\phi \rangle\big\rvert.
\end{split}
\end{equation}
The last term vanishes as $R \to 0$. That is, the error can be controlled by increasing the number of particles $J$ and decreasing $R$. Though a word of caution is appropriate, as particle filters can become ill-behaved in practice if the likelihoods are too narrow (too small $R$). However, this also depends on the quality of the proposal distribution.

Lastly, while Theorem \ref{thm:m2convergence} is only valid if $\rho(x_{n+1},x_n)$ is bounded, this can be ensured by either inflating the covariance of the proposal distribution or replacing the Gaussian proposal with a Student's t proposal \citep[Chapter 9]{Cappe2005}. 

\section{A Stability Result for Gaussian Filters}
ODE solvers are often characterised by the properties of their solution to the linear test equation
\begin{equation}\label{eq:linear_test_equation}
\dot{y}(t) = \lambda y(t), \ y(0) = 1,
\end{equation}
where $\lambda$ is some complex number. A numerical solver is said to be \emph{A-stable} if the approximate solution tends
to zero for any fixed step size $h$ whenever the real part of $\lambda$ resides in the left half-plane \citep{Dahlquist1963}. Recall that if $y_0 \in \mathbb{R}^d$ and $\Lambda \in \mathbb{R}^{d\times d}$ then the ODE $\dot{y}(t) = \Lambda y(t), \ y(0) = y_0$ is said to be asymptotically stable if $\lim_{t \to \infty} y(t) = 0$, which is precisely when the real part of eigenvalues of $\Lambda$ are in the left half-plane. That is, A-stability is the notion that a numerical solver preserves asymptotic stability of linear time-invariant ODEs.

While the present solvers are not designed to solve complex valued ODEs, a real system equivalent to Eq. \eqref{eq:linear_test_equation} is given by
\begin{equation}\label{eq:linear_test_equation_real}
\dot{y}(t)  = \Lambda_{\text{test}} y(t), \quad y^\T(0) = [1\  0], \ 
\end{equation}
where $\lambda = \lambda_1 + i \lambda_2$ and
\begin{equation}\label{eq:complex_test_matrix}
\Lambda_{\text{test}} = \begin{bmatrix} \lambda_1 & - \lambda_2 \\ \lambda_2 & \lambda_1  \end{bmatrix}.
\end{equation}
However, to leverage classical stability results from the theory of Kalman filtering we investigate a slightly different test equation, namely
\begin{equation}
\dot{y}(t) = \Lambda y(t),\ y(0) = y_0,
\end{equation}
where $\Lambda \in \mathbb{R}^{d\times d}$ is of full rank. In this case Eqs. \eqref{eq:kalman_prediction} and \eqref{eq:affine_gaussian_update} give the following recursion for $\mu_n^P$
\begin{subequations}
\begin{align}
\mu_{n+1}^P &= (A(h) - A(h) K_n H) \mu_n^P, \\
\mu_n^F &= (\mathrm{I} - K_n H)\mu_n^P,
\end{align}
\end{subequations}
where we recall that $H = \dot{C}- C\Lambda$ and $z_n = 0$. If there exists a limit gain $\lim_{n\to \infty} K_n = K_\infty$ then asymptotic stability of the filter holds provided that the eigenvalues of $(A(h) - A(h) K_\infty H)$ are strictly within the unit circle \citep[Appendix C, page 341]{anderson1979}. That is, $\lim_{n\to\infty} \mu_n^P = 0$ and as a direct consequence $\lim_{n\to\infty} \mu_n^F = 0$.

We shall see that the Kalman filter using an IWP($q$) prior is asymptotically stable. For the IWP($q$) process on $\mathbb{R}^d$ we have $u = 0$, $L = \mathrm{e}_{q+1}\otimes \Gamma^{1/2}$, and $F = (\sum_{i=1}^q \mathrm{e}_i\mathrm{e}_{i+1}^\T)\otimes \mathrm{I}$, where $\mathrm{e}_i \in \mathbb{R}^d$ is the $i$th canonical eigenvector, $\Gamma^{1/2}$ is the symmetric square root of some positive semi-definite matrix $\Gamma \in \mathbb{R}^{d\times d}$, $\mathrm{I} \in \mathbb{R}^{d\times d}$ is the identity matrix, and $\otimes$ is Kronecker's product. By using Eq. \eqref{eq:lti_disc}, the properties of Kronecker products, and the definition of the matrix exponential the equivalent discrete-time system is given by
\begin{subequations}\label{eq:lti_disc_iwp_nd} 
\begin{align}
A(h) &= A^{(1)}(h) \otimes \mathrm{I},  \label{eq:lti_disc_iwp_nd_transition}\\
\xi(h) &= 0, \\
Q(h) &= Q^{(1)}(h) \otimes \Gamma, \label{eq:lti_disc_iwp_nd_cov}
\end{align}
\end{subequations}
where $A^{(1)}(h) \in \mathbb{R}^{(q+1)\times (q+1)}$ and $Q^{(1)}(h) \in \mathbb{R}^{(q+1)\times (q+1)}$ are given by \citep[Appendix A]{Kersting2018}\footnote{Note that \cite{Kersting2018} uses indexing $i,j=0,\ldots,q$ while we here use $i,j=1,\ldots,q+1$.}

\begin{subequations}\label{eq:lti_disc_iwp_1d}
\begin{align}
A_{ij}^{(1)}(h) &= \mathbb{I}_{i \leq j} \frac{h^{j-i}}{(j-i)!}, \label{eq:lti_disc_iwp_1d_transition} \\
Q_{ij}^{(1)}(h) &=  \frac{ h^{2q+3-i-j}}{(2q+3-i-j)(q+1-i)!(q+1-j)!}, \label{eq:lti_disc_iwp_1d_cov}
\end{align}
\end{subequations}
and $\mathbb{I}_{i \leq j}$ is an indicator function.  
Before proceeding we need to introduce the notions of stabilisability and detectability from Kalman filtering theory. These notions can be found in \citet[Appendix C]{anderson1979}.

\begin{definition}[Complete Stabilisability]
The pair $[A,G]$ is completely stabilisable if $w^\T G = 0$ and $w^\T A = \eta w^\T$ for some constant $\eta$ implies $\abs{\eta} < 1$ or $w = 0$.
\end{definition}

\begin{definition}[Complete Detectability]\label{def:complete_detectability}
\footnote{
\citet{anderson1979} denotes the measurement matrix by $H^\T$ while we denote it by $H$. With this in mind our notion of complete detectability does not differ from \citet{anderson1979}.
}
$[A,H]$ is completely detectable if $[A^\T,H^\T]$ is completely stabilisable.
\end{definition}

Before we state the stability result of this section the following two lemmas are useful.

\begin{lemma}\label{lem:process_cov_rank}
Consider the discretised IWP($q$) prior on $\mathbb{R}^d$ as given by Eq. \eqref{eq:lti_disc_iwp_nd}. Let $h > 0$ and $\Gamma$ be positive definite. Then, the $d\times d$ blocks of $Q(h)$, denoted by $Q_{i,j}(h),\ i,j = 1,2,\ldots,q+1$ are of full rank.
\end{lemma}
\begin{proof}
From Eq. \eqref{eq:lti_disc_iwp_nd_cov} we have $Q_{i,j}(h) = Q_{i,j}^{(1)}(h) \Gamma$. From Eq. \eqref{eq:lti_disc_iwp_1d_cov} and $h> 0$ we have $Q_{i,j}^{(1)}(h) > 0$, and since $\Gamma$ is positive definite it is of full rank. It then follows that $Q_{i,j}(h)$ is of full rank as well. \qed
\end{proof}

\begin{lemma}\label{lem:eigspace}
Let $A(h)$ be the transition matrix of an IWP($q$) prior as given by Eq. \eqref{eq:lti_disc_iwp_nd_transition} and $h>0$, then $A(h)$ has a single eigenvalue given by $\eta = 1$. Furthermore, the right-eigenspace is given by
\begin{equation*}
\operatorname{span}[\mathrm{e}_1,\mathrm{e}_2,\ldots,\mathrm{e}_d],
\end{equation*}
where $\mathrm{e}_i \in \mathbb{R}^{(q+1)d}$ are canonical basis vectors, and the left-eigenspace is given by
\begin {equation*}
\operatorname{span}[\mathrm{e}_{qd+1},\mathrm{e}_{qd+2},\ldots,\mathrm{e}_{(q+1)d}].
\end{equation*}
\end{lemma}
\begin{proof}
Firstly, from Eqs. \eqref{eq:lti_disc_iwp_nd_transition} and \eqref{eq:lti_disc_iwp_1d_transition} it follows that $A(h)$ is block upper-triangular with identity matrices on the block diagonal, hence the characteristic equation is given by
\begin{equation}
\det(A(h)-\eta\mathrm{I}) = (1-\eta)^{(q+1)d} = 0,
\end{equation}
we conclude that the only eigenvalue is $\eta = 1$. To find the right-eigenspace let $w^\T = [w_1^\T,w_2^\T,\ldots,w_{q+1}^\T]$, $w_i \in \mathbb{R}^d,\ i=1,2,\ldots,q+1$ and solve $A(h)w = w$, which by using Eqs. \eqref{eq:lti_disc_iwp_nd_transition} and \eqref{eq:lti_disc_iwp_1d_transition} can be written as
\begin{equation}
(A(h)w)_l = \sum_{r=0}^{q+1-l} \frac{h^r}{r!} w_{r+l}, \quad l = 1,2,\ldots,q+1,
\end{equation}
where $(\cdot)_l$ is the $l$th sub-vector of dimension $d$. Starting with $l = q+1$ we trivially have $w_{q+1} = w_{q+1}$.
For $l = q$ we have $w_q + w_{q+1}h = w_q$ but $h > 0$, hence $w_{q+1} = 0$. Similarly for $l = q-1$ we have $w_{q-1} = w_{q-1} + w_qh + w_{q+1} h^2/2 = w_{q-1} + w_q h + 0 \cdot h^2/2$. Again since $h > 0$ we have $w_q = 0$. By repeating this argument we have $w_1 = w_1$ and $w_i = 0,\ i = 2,3,\ldots,q+1$.  Therefore all eigenvectors $w$ are of the form $w^\T = [w_1^\T,0^\T,\ldots,0^\T] \in \operatorname{span}[\mathrm{e}_1,\mathrm{e}_2,\ldots,\mathrm{e}_d]$. Similarly, for the left eigenspace we have 
\begin{equation}
(w^\T A(h))_l = \sum_{r=0}^{l-1} \frac{h^r}{r!} w_{l-r}^\T, \quad l = 1,2,\ldots,q+1.
\end{equation}
Starting with $l = 1$ we have trivially that $w_1^\T = w_1^\T$. For $l = 2$ we have $w_2^\T + w_1^\T h = w_2^\T$ but $h>0$, hence $w_1=0$. For $l = 3$ we have $w_3^\T = w_3^\T + w_2^\T h + w_1^\T h^2/2 = w_3^\T + w_2^\T h + 0^\T\cdot h^2/2$ but $h > 0$ hence $w_2 =0$. By repeating this argument we have $w_i= 0, i=1,\ldots,q$ and $w_{q+1} = w_{q+1}$. Therefore, all left eigenvectors are of the form 
$w^\T = [0^\T,\ldots,0^\T,w_{q+1}^\T] \in \operatorname{span}[\mathrm{e}_{qd+1},\mathrm{e}_{qd+2},\ldots,\mathrm{e}_{(q+1)d}]$. 
\qed
\end{proof}
We are now ready to state the main result of this section. Namely, that the Kalman filter that produces exact inference in Eq. \eqref{eq:inference_problem} for linear vector fields is asymptotically stable if the linear vector field is of full rank.
\begin{theorem}\label{thm:stability}
Let $\Lambda \in \mathbb{R}^{d\times d}$ be a matrix with full rank and consider the linear ODE
\begin{equation}\label{eq:thmtest}
\dot{y}(t) = \Lambda y(t).
\end{equation}
Consider estimating the solution of Eq. \eqref{eq:thmtest} using an IWP(q) prior with the same conditions on $\Gamma$ as in Lemma \ref{lem:process_cov_rank}. Then the Kalman filter estimate of the solution to Eq. \eqref{eq:thmtest} is asymptotically stable.
\end{theorem}

\begin{proof}
From Eq. \eqref{eq:inference_problem} we have that the Kalman filter operates on the following system
\begin{subequations}
\begin{align}
X_{n+1} &= A(h) X_n + Q^{1/2}(h)W_{n+1}, \\
Z_n &= H X_n,
\end{align}
\end{subequations}
where $H = [-\Lambda, \mathrm{I},0,\ldots,0]$ and $W_n$ are i.i.d. standard Gaussian vectors. It is sufficient to show that $[A(h),H]$ is completely detectable and $[A(h),Q^{1/2}(h)]$ is completely stabilisable
\citep[Chapter 4, page 77]{anderson1979}. We start by showing complete detectability. If we let $w^\T = [w_1^\T,\ldots,w_{q+1}^\T]$, $w_i \in \mathbb{R}^d,\ i=1,2,\ldots,q+1$, then by Lemma \ref{lem:eigspace} we have that $w^\T A^\T(h) = \eta w^\T$ for some $\eta$ implies that either
$w = 0$ or $w^\T = [w_1^\T,0^\T,\ldots,0^\T]$ for some $w_1 \in \mathbb{R}^d$ and $\eta = 1$. Furthermore, $w^\T H^\T = -w_1^\T \Lambda^\T + w_2^\T = 0$ implies that $w_2 = \Lambda w_1$. However, by the previous argument, we have $w_2 = 0$, therefore $0 = \Lambda w_1$ but $\Lambda$ is full rank by assumption so $w_1 = 0$. Therefore, $[A^\T(h),H^\T]$ is completely detectable. As for complete stabilisability, again by Lemma \ref{lem:eigspace}, we have $w^\T A(h) = \eta w^\T$ for some $\eta$, which implies either $w = 0$ or $w^\T = [0^\T,\ldots,0^\T,w_{q+1}^\T]$ and $\eta = 1$. Furthermore, since the nullspace of $Q^{1/2}(h)$ is the same as the nullspace of $Q(h)$, we have that \sloppy{$w^\T Q^{1/2}(h) = 0$} is equivalent to $w^\T Q(h) = 0$, which is given by
\begin{equation*}
w^\T Q(h) = \begin{bmatrix} w_{q+1}^\T Q_{q+1,1}(h) & \ldots & w_{q+1}^\T Q_{q+1,q+1}(h) \end{bmatrix} = 0,
\end{equation*}
but by Lemma \ref{lem:process_cov_rank} the blocks  $Q_{i,j}(h)$ have full rank so $w_{q+1} = 0$ and thus $w = 0$. To conclude, we have that $[A(h),Q^{1/2}(h)]$ is completely stabilisable and $[A(h),H]$ is completely detectable and therefore the Kalman filter is asymptotically stable.\qed
\end{proof}

\begin{corollary}\label{cor:stability_ekfukf}
In the same setting as Theorem \ref{thm:stability}, the EKF and UKF are asymptotically stable.
\end{corollary}

\begin{proof}
Since the vector field is linear and therefore affine Lemma \ref{lem:ekfukf_affine} implies that EKF and UKF reduce to the exact Kalman filter, which is asymptotically stable by Theorem \ref{thm:stability}.\qed
\end{proof}

It is worthwhile to note that $\Lambda_{\text{test}}$ is of full rank for all $[\lambda_1\ \lambda_2]^\T \in \mathbb{R}^2 \setminus \{ 0 \}$, and consequently Theorem \ref{thm:stability} and Corollary \ref{cor:stability_ekfukf} guarante A-stability for the EKF and UKF in the sense of \citet{Dahlquist1963}\footnote{Some authors require stability on the line $\lambda_1 = 0$ as well \citep{Hairer1996}. Due to the exclusion of origin EKF and UKF cannot be said to be A-stable in this sense.}. Lastly, a peculiar fact about Theorem \ref{thm:stability} is that it makes no reference to the eigenvalues of $\Lambda$ (i.e. the stability properties of the ODE). That is, the Kalman filter will be asymptotically stable even if the underlying ODE is not, provided that, $\Lambda$ is of full rank. This may seem awkward but it is rarely the case that the ODE that we want to integrate is unstable, and even in such a case most solvers will produce an error that grows without a bound as well. Though all of the aforementioned properties are at least partly consequences of using IWP($q$) as a prior and they may thus be altered by changing the prior.  

\section{Uncertainty Calibration}
In practice the model parameters, $(F,u,L)$, might depend on some parameters that need to be estimated for the probabilistic solver to report
appropriate uncertainty in the estimated solution to Eq. \eqref{eq:ode}. The diffusion matrix $L$ is of particular importance as it determines
the gain of the Wiener process entering the system in Eq. \eqref{eq:ode_prior} and thus determines how \emph{'diffuse'} the prior is.
Herein we shall only concern ourselves with estimating $L$, though, one might anticipate future interest in estimating $F$ and $u$ as well. However, let us start with a few words on the monitoring of errors in numerical solvers in general.

\subsection{Monitoring of Errors in Numerical Solvers}
An important aspect of numerical analysis is to monitor the error of a method. While the goal of probabilistic solvers is to do so by calibration of a probabilistic model, the approach of classical numerical analysis is to examine the local and global errors. The global error can be bounded but is typically impractical for monitoring error \cite[Chapter II.3]{hairer87:_solvin_ordin_differ_equat_i}. A more practical approach is to monitor (and control) the accumulation of local errors. This can be done by using two step sizes together with Richardson extrapolation \cite[Theorem 4.1]{hairer87:_solvin_ordin_differ_equat_i}. Though, perhaps more commonly this is done via embedded Runge--Kutta methods \cite[Chapter II.4]{hairer87:_solvin_ordin_differ_equat_i} or the Milne device \cite{Byrne1975}. 

In the context of filters, the relevant object in this regard is the scaled residual $S_n^{-1/2}(z_n - \hat{z}_n)$. Due to its role in the prediction-error decomposition, which is defined below, it directly monitors the calibration of the predictive distribution. \cite{schober2018} showed how to use this quantity to effectively control step sizes in practice. It was also recently shown in \cite[Section 7]{Kersting2018}, that in the case of $q=1$, fixed $\sigma^2$ (amplitude of the Wiener process) and Integrated Wiener Process prior, the posterior standard deviation computed by the solver of \cite{schober2018} contracts at the same rate as the worst-case error as the step size goes to zero---thereby preventing both under- and over-confidence.

In the following we discuss effective strategies for calibrating $L$ when it is given by $L = \sigma \breve{L}$ for fixed $\breve{L}$ thus providing a probabilistic quantification of the error in the proposed solvers.

\subsection{Uncertainty Calibration for Affine Vector Fields}\label{susbec:affine_calibration}
As noted in Section \ref{subsec:affine_vector_fields}, the Kalman filter produces the exact solution to the inference problem in Eq. \eqref{eq:inference_problem} when the vector field is affine. Furthermore, the marginal likelihood $p(z_{1:N})$ can be computed during the execution of the Kalman filter by the prediction error decomposition \citep{Schweppe1965}, which is given by:
\begin{equation}\label{eq:marginal_likelihood}
\begin{split}
p(z_{1:N}) &= p(z_1)\prod_{n=2}^N p(z_n\mid z_{1:n-1}) \\
&= \prod_{n=1}^N \mathcal{N}(z_n; \hat{z}_n,S_n).
\end{split}
\end{equation}
While the marginal likelihood in Eq. \eqref{eq:marginal_likelihood} is certainly straightforward to compute without adding much computational cost, maximising it is a different story in general. In the particular case when the diffusion matrix $L$ and the initial covariance $\Sigma_0$ are given by re-scaling fixed matrices $L = \sigma\breve{L}$ and $\Sigma_0 = \sigma^2\breve{\Sigma}_0$ for some scalar $\sigma > 0$, then uncertainty calibration can be done by a simple post-processing step after running the Kalman filter, as is shown in Proposition \ref{prop:s2mle_exact} below.
\begin{proposition}\label{prop:s2mle_exact}
Let $f(y,t) = \Lambda(t)y + \zeta(t)$, $\Sigma_0 = \sigma^2 \breve{\Sigma}_0$, $L = \sigma \breve{L}$, $R = 0$ and denote the equivalent discrete-time process noise covariance for the prior model $(F,u,\breve{L})$ by $\breve{Q}(h)$. Then the Kalman filter estimate to the solution of
\begin{equation*}
\dot{y}(t) = f(y(t),t)
\end{equation*}
that uses the parameters $(\mu_0^F,\Sigma_0,A(h),\xi(h),Q(h))$ is equal to the Kalman filter estimate that uses the parameters $(\mu_0^F,\breve{\Sigma}_0,A(h),\xi(h),\breve{Q}(h))$. More specifically, if we denote the filter mean and covariance at time $n$ using the former parameters by $(\mu_n^F,\Sigma_n^F)$ and the corresponding filter mean and covariance using the latter parameters by $(\breve{\mu}_n^F,\breve{\Sigma}_n^F)$, then $(\mu_n^F,\Sigma_n^F) = (\breve{\mu}_n^F,\sigma^2\breve{\Sigma}_n^F)$. Additionally, denote the predicted mean and covariance of the measurement $Z_n$ by $\breve{z}_n$ and $\breve{S}_n$, respectively, when using the parameters $(\mu_0^F,\breve{\Sigma}_0,A(h),\xi(h),\breve{Q}(h))$. Then the maximum likelihood estimate of $\sigma^2$, denoted by $\widehat{\sigma^2_N}$, is given by
\begin{equation}\label{eq:s2mle_exact}
\widehat{\sigma^2_N} = \frac{1}{Nd} \sum_{n=1}^N (z_n - \breve{z}_n)^\T \breve{S}_n^{-1} (z_n - \breve{z}_n).
\end{equation}
\end{proposition}

Proposition \ref{prop:s2mle_exact} is just an amalgamation of statements from \citet{Tronarp2019a}. Nevertheless, we provide an accessible proof in Appendix \ref{appendix:proof4}.

\subsection{Uncertainty Calibration for Non-Affine Vector Fields}\label{susbec:non_affine_calibration}
For non-affine vector fields the issue of parameter estimation becomes more complicated. The Bayesian filtering problem is not solved exactly and consequently any marginal likelihood will be approximate as well. Nonetheless, a common approach in the Gaussian filtering framework is to approximate the marginal likelihood in the same manner as the filtering solution is approximated \citep[Chapter 12.3.3]{Sarkka2013},  that is:
\begin{equation}\label{eq:marginal_likelihood_approximation}
p(z_{1:N}) \approx \prod_{n=1}^N \mathcal{N}(z_n; \hat{z}_n,S_n),
\end{equation}
where $\hat{z}_n$ and $S_n$ are the quantities in Eq. \eqref{eq:general_gaussian_update} approximated by some method (e.g. EKF). Maximising Eq. \eqref{eq:marginal_likelihood_approximation} is a common approach in signal processing \citep{Sarkka2013} and referred to as \emph{quasi maximum likelihood} in time series literature \citep{Lindstrom2015}. Both Eq. \eqref{eq:schober} and Eq. \eqref{eq:ekf} can be thought of as Kalman updates for the case where the vector field is approximated by a piece-wise affine function, without modifying $\Sigma_0$, $Q(h)$, and $R$. For instance the affine approximation of the vector field due to the EKF on the discretisation interval $[t_n,t_{n+1})$ is given by
\begin{subequations}
\begin{align}
\hat{\zeta}_n(t) &= f\big(C\mu_n^P,t_n\big) - J_f\big(C\mu_n^P,t_n\big)C\mu_n^P, \\
\hat{\Lambda}_n(t) &= J_f\big(C\mu_n^P,t_n\big),\\
\hat{f}_n(y,t) &= \hat{\Lambda}_n(t)y + \hat{\zeta}_n(t).
\end{align}
\end{subequations}
While the vector field is approximated by a piece-wise affine function, the discrete-time filtering problem Eq. \eqref{eq:inference_problem} is still simply an affine problem, without modifications of $\Sigma_0$, $Q(h)$, and $R$. Therefore, the results of Proposition \ref{prop:s2mle_exact} still apply and the $\sigma^2$ maximising the approximate marginal likelihood in Eq. \eqref{eq:marginal_likelihood_approximation} can be computed in the same manner as in Eq. \eqref{eq:s2mle_exact}.

On the other hand, it is clear that dependence on $\sigma^2$ in Eq. \eqref{eq:general_gaussian_update} is non-trivial in general, which is also true for the quadrature approaches of Section \ref{subsec:quadrature}. Therefore, maximising Eq. \eqref{eq:marginal_likelihood_approximation} for the quadrature approaches is not as straightforward. However, by Taylor series expanding the vector field in Eq. \eqref{eq:general_gaussian_update} one can see that the numerical integration approaches are roughly equal to the Taylor series approaches provided that $\breve{\Sigma}_n^P$ is small. Therefore, we opt for plugging in the corresponding quantities from the quadrature approximations into Eq. \eqref{eq:s2mle_exact} in order to achieve computationally cheap calibration of these approaches.

\begin{remark}
A local calibration method for $\sigma^2$ is given by \cite[Eq. (45)]{schober2018}, which in fact corresponds to an $h$-dependent prior, with the diffusion matrix in Eq. \eqref{eq:ode_prior} $L = L(t)$ being piece-wise constant over integration steps. Moreover, \citet{schober2018} had to neglect the dependence of $\Sigma_n^P$ on the likelihood. Here we prefer the estimator given in Eq. \eqref{eq:s2mle_exact} since it is attempting to maximise the likelihood from the globally defined probability model in Eq. \eqref{eq:inference_problem}, and it succeeds for affine vector fields.
\end{remark}

More advanced methods for calibrating the parameters of the prior can be developed by combining the Gaussian smoothing equations \citep[Chapter 10]{Sarkka2013} with the expectation maximisation method \citep{Kokkala2014} or variational Bayes \citep{Taniguchi2017}.

\subsection{Uncertainty Calibration of Particle Filters}
If calibration of Gaussian filters was complicated by having a non-affine vector field, the situation for particle filters is even more challenging. There is, to the authors' knowledge, no simple estimator of the scale of the Wiener process (such as Proposition \ref{prop:s2mle_exact}) even for the case of affine vector fields. However, the literature on parameter estimation using particle methods is vast so we proceed to point the reader towards some alternatives. In the class of off-line methods, \citet{Schon2011} uses a particle smoother to implement an expectation maximisation algorithm, while \citet{Lindsten2013} uses a particle Markov chain Monte Carlo methods to implement a stochastic approximation expectation maximisation algorithm. One can also use the iterated filtering method of \citet{Ionides2011} to get a maximum likelihood estimator, or particle Markov chain Monte Carlo \citep{Andrieu2010}. 

On the other hand, if on-line calibration is required then the gradient based recursive maximum likelihood estimator by \citet{Doucet2003} can be used, or the on-line version of iterated filtering by \citet{Lindstrom2012}. Furthermore, \citet{Storvik2002} provides an alternative for on-line calibration when sufficient statics of the parameters are finite dimensional and can be computed recursively in $n$. An overview on parameter estimation using particle filters was also given by \citet{Kantas2009}.

\section{Experimental Results}\label{sec:experiments}
In this section we evaluate the different solvers presented in this paper in different scenarios. Though before we proceed to the experiments we define some summary metrics with which assessments of accuracy and uncertainty quantification can be made. The root mean square error (RMSE) is often used to assess accuracy of filtering algorithms and is defined by
\begin{equation*}
\mathrm{RMSE} = \sqrt{\frac{1}{N} \sum_{n=1}^N \norm{y(nh) - C \mu^F_n}^2 }.
\end{equation*}
In fact $y(nh) - C \mu^F_n$ is precisely the \emph{global error} at time $t_n$ \citep[Eq. (3.16)]{hairer87:_solvin_ordin_differ_equat_i}.
As for assessing the uncertainty quantification, the $\chi^2$-statistics is commonly used \citep{BarShalom2004}. That is, in a linear Gaussian model the following quantities 
\begin{equation*}
\big(y(nh) - C\mu_n^F\big)^\T [C \Sigma^F_n C^\T]^{-1} \big(y(nh) - C\mu_n^F\big), \ n=1,\ldots,N,
\end{equation*}
are i.i.d. $\chi^2(d)$. For a trajectory summary we define the average $ \chi^2$-statistics as
\begin{equation*}
\bar{\chi^2} = \frac{1}{N}\sum_{n=1}^N \big(y(nh) - C\mu_n^F\big)^\T [C \Sigma^F_n C^\T]^{-1} \big(y(nh) - C\mu_n^F\big).
\end{equation*}
For an accurate and well calibrated model the RMSE is small and $\bar{\chi^2} \approx d$. In the succeeding discussion we shall refer to a method producing $\bar{\chi^2} < d$ or $\bar{\chi^2} > d$ as \emph{underconfident} or \emph{overconfident}, respectively.

\subsection{Linear Systems}\label{subsec:experiment:linear}
In this experiment we consider a linear system given by
\begin{subequations}
\begin{align}
\Lambda &= \begin{bmatrix} \lambda_1 & -\lambda_2 \\ \lambda_2 & \lambda_1 \end{bmatrix},\\
\dot{y}(t) &= \Lambda y(t), \quad y(0) = \mathrm{e}_1.
\end{align}
\end{subequations}
This makes for a good test model as the inference problem in Eq. \eqref{eq:inference_problem} can be solved exactly, and consequently its adequacy can be assessed. We compare exact inference by the Kalman filter (KF)\footnote{Again note that the EKF and appropriate numerical quadrature methods are equivalent to this estimator here (see Lemma \ref{lem:ekfukf_affine}).}  (see Section \ref{subsec:affine_vector_fields}) with the approximation due to \citet{schober2018} (SCH) (see Proposition \ref{prop:schober}) and the covariance approximation due to \citet{Kersting2016UAI} (KER) (see Proposition \ref{prop:kersting}). The integration interval is set to $[0,10]$ and all methods use an IWP($q$) prior for $q=1,2,\ldots,6$, and the initial mean is set to $\E[X^{(j)}(0)] = \Lambda^{j-1}y(0)$ for $j=1,\ldots,q+1$, with variance set to zero (exact initialisation). The uncertainty of the methods is calibrated by the maximum likelihood method (see Proposition \ref{prop:s2mle_exact}), and the methods are examined for 10 step sizes uniformly placed on the interval $[10^{-3},10^{-1}]$.

We examine the parameters $\lambda_1 = 0$ and $\lambda_2 = \pi$ (half a revolution per unit of time with no damping). The RMSE is plotted against step size in Figure \ref{fig:experiment_linear1:rmse}. It can be seen that SCH is a slightly better than KF and KER for $q = 1$ and small step sizes, and KF becomes slightly better than SCH for large step size while KER becomes significantly worse than both KF and SCH.
For $q > 1$, it can be seen that the RMSE is significantly lower for KF than for SCH/KER in general with performance differing between one and two orders of magnitude. Particularly, the superior stability properties of KF are demonstrated (see Theorem \ref{thm:stability}) for $q > 3$ where both SCH and KER produce massive errors for larger step sizes.

Furthermore, the average $\chi^2$-statistic is shown in Figure \ref{fig:experiment_linear1:chi2}. All methods appear to be overconfident for $q=1$ with SCH performing best, followed by KER. On the other hand, for $1 < q < 5$, SCH and KER remain overconfident for the most part, while KF is underconfident. Our experiments also show that unsurprisingly all methods perform better for smaller $\lvert\lambda_2\rvert$ (frequency of the oscillation). However, we omit visualising this here.

Finally, a demonstration of the error trajectory for the first component of $y$ and the reported uncertainty of the solvers is shown in Figure \ref{fig:demonstration_linear1} for $h = 10^{-2}$ and $q = 2$. Here it can be seen that all methods produce similar errors bars, though SCH and KER produce errors that oscillate far outside their reported uncertainties.

\begin{figure}[t!]
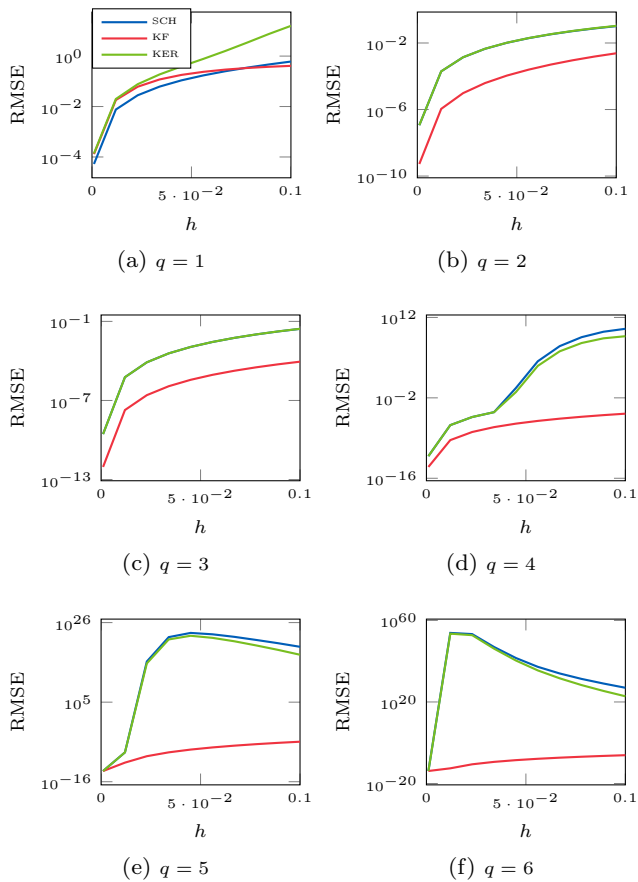

\centering
\subfloat[\scriptsize{$q = 1$}\label{fig:experiment_linear1_rmse_q1}]{\input{fig/experiment_linear1/experiment_linear1_rmse_q1.tex}}
\subfloat[\scriptsize{$q = 2$}\label{fig:experiment_linear1_rmse_q2}]{\input{fig/experiment_linear1/experiment_linear1_rmse_q2.tex}}\\
\subfloat[\scriptsize{$q = 3$}\label{fig:experiment_linear1_rmse_q3}]{\input{fig/experiment_linear1/experiment_linear1_rmse_q3.tex}}
\subfloat[\scriptsize{$q = 4$}\label{fig:experiment_linear1_rmse_q4}]{\input{fig/experiment_linear1/experiment_linear1_rmse_q4.tex}}\\
\subfloat[\scriptsize{$q = 5$}\label{fig:experiment_linear1_rmse_q5}]{\input{fig/experiment_linear1/experiment_linear1_rmse_q5.tex}}
\subfloat[\scriptsize{$q = 6$}\label{fig:experiment_linear1_rmse_q6}]{\input{fig/experiment_linear1/experiment_linear1_rmse_q6.tex}}
\caption{RMSE of KF, SCH, and KER on the undamped oscillator using IWP($q$) priors for $q=1,\ldots,6$ plotted against step size.}\label{fig:experiment_linear1:rmse}
\end{figure}

\begin{figure}[t!]
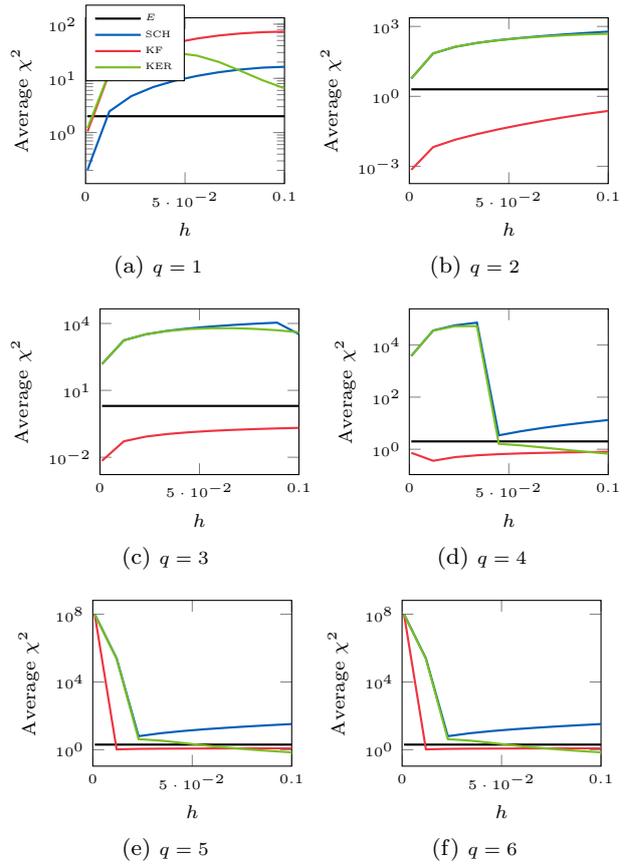

\centering
\subfloat[\scriptsize{$q = 1$}\label{fig:experiment_linear1_chi2_q1}]{\input{fig/experiment_linear1/experiment_linear1_chi2_q1.tex}}
\subfloat[\scriptsize{$q = 2$}\label{fig:experiment_linear1_chi2_q2}]{\input{fig/experiment_linear1/experiment_linear1_chi2_q2.tex}}\\
\subfloat[\scriptsize{$q = 3$}\label{fig:experiment_linear1_chi2_q3}]{\input{fig/experiment_linear1/experiment_linear1_chi2_q3.tex}}
\subfloat[\scriptsize{$q = 4$}\label{fig:experiment_linear1_chi2_q4}]{\input{fig/experiment_linear1/experiment_linear1_chi2_q4.tex}}\\
\subfloat[\scriptsize{$q = 5$}\label{fig:experiment_linear1_chi2_q5}]{\input{fig/experiment_linear1/experiment_linear1_chi2_q5.tex}}
\subfloat[\scriptsize{$q = 6$}\label{fig:experiment_linear1_chi2_q6}]{\input{fig/experiment_linear1/experiment_linear1_chi2_q6.tex}}
\caption{Average $\chi^2$-statistic of KF, SCH, and KER on the undamped oscillator using IWP($q$) priors for $q=1,\ldots,6$ plotted against step size. The expected $\chi^2$-statistic is shown in black (E).}\label{fig:experiment_linear1:chi2}
\end{figure}

\begin{figure}[t!]
\centering
\input{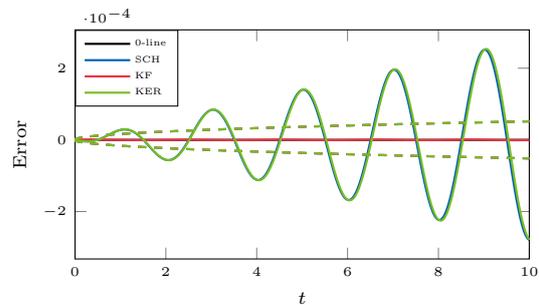}
\caption{The errors (solid lines) and $\pm$ 2 standard deviation bands (dashed) for KF, SCH, and KER on the undamped oscillator with $q=2$ and $h = 10^{-2}$. A line at 0 is plotted in solid black.}\label{fig:demonstration_linear1}
\end{figure}

\subsection{The Logistic Equation}\label{subsec:experiment:logistic}
In this experiment the logistic equation is considered:
\begin{equation}\label{eq:experiment:logistic}
\dot{y}(t) = r y(t) \left(1 - y(t)\right), \quad y(0) = 1\cdot 10^{-1},
\end{equation}
which has the solution:
\begin{equation}
y(t) = \frac{\exp(rt)}{1/y_0 -1 + \exp(rt)}.
\end{equation}
In the experiments $r$ is set to $r = 3$. We compare the zeroth order solver (Proposition \ref{prop:schober}) \citep{schober2018} (SCH), the first order solver in Eq. \eqref{eq:ekf} (EKF), a numerical integration solver based on the covariance approximation in Proposition \ref{prop:kersting} \citep{Kersting2016UAI} (KER), and a numerical integration solver based on approximating Eq. \eqref{eq:general_gaussian_update} (UKF). Both numerical integration approaches use a third degree fully symmetric rule \citep[see][]{McNameeStenger1967}. The integration interval is set to $[0,2.5]$ and all methods use an IWP($q$) prior for $q=1,2,\ldots,4$, and the initial mean of $X^{(1)}$, $X^{(2)}$, and $X^{(3)}$ are set to $y(0)$, $f(y(0))$, and $J_f(y(0))f(y(0))$, respectively (correct values), with zero covariance. The remaining state components $X^{(j)}, j>3$ are set to zero mean with unit variance. The uncertainty of the methods is calibrated by the quasi maximum likelihood method as explained in Section \ref{susbec:non_affine_calibration}, and the methods are examined for 10 step sizes uniformly placed on the interval $[10^{-3},10^{-1}]$.

The RMSE is plotted against step size in Figure \ref{fig:experiment_logistic1:rmse}. It can be seen that EKF and UKF tend to produce smaller errors by more than an order of magnitude than SCH and KER in general, with the notable exception of the UKF behaving badly for small step sizes and $q = 4$. This is probably due to numerical issues for generating the integration nodes, which requires the computation of matrix square roots \citep{Julier2000} that can become inaccurate for ill-conditioned matrices. Additionally, the average $\chi^2$-statistic is plotted against step size in Figure \ref{fig:experiment_logistic1:chi2}. Here it appears that all methods tend to be underconfident for $q=1,2$, while SCH becomes overconfident for $q=3,4$.

A demonstration of the error trajectory and the reported uncertainty of the solvers is shown in Figure \ref{fig:demonstration_linear1} for $h = 10^{-1}$ and $q = 2$. SCH and KER produce similar errors and they are hard to discern in the figure. The same goes for EKF and UKF. Additionally, it can be seen that the solvers produce qualitatively different uncertainty estimates. While the uncertainty of EKF and UKF first grows to then shrink as the the solution approaches the fixed point at $y(t) = 1$, the uncertainty of SCH grows over the entire interval with the uncertainty of KER growing even faster.

\begin{figure}[t!]
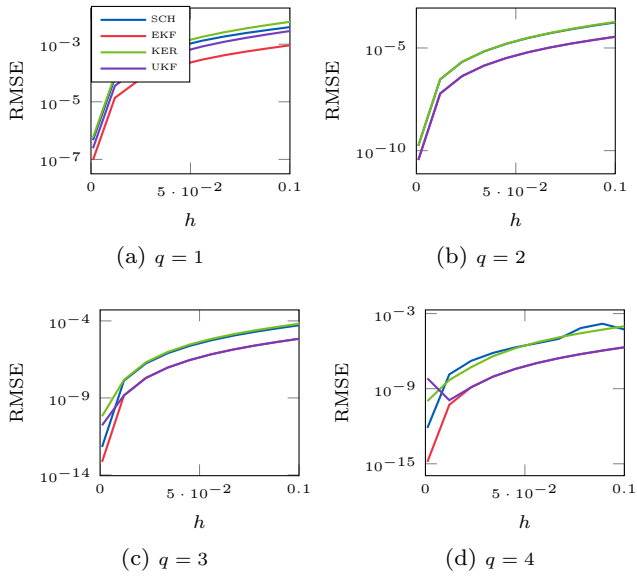

\centering
\subfloat[\scriptsize{$q = 1$}\label{fig:experiment_logistic1_rmse_q1}]{\input{fig/experiment_logistic1/experiment_logistic1_rmse_q1.tex}}
\subfloat[\scriptsize{$q = 2$}\label{fig:experiment_logistic1_rmse_q2}]{\input{fig/experiment_logistic1/experiment_logistic1_rmse_q2.tex}}\\
\subfloat[\scriptsize{$q = 3$}\label{fig:experiment_logistic1_rmse_q3}]{\input{fig/experiment_logistic1/experiment_logistic1_rmse_q3.tex}}
\subfloat[\scriptsize{$q = 4$}\label{fig:experiment_logistic1_rmse_q4}]{\input{fig/experiment_logistic1/experiment_logistic1_rmse_q4.tex}}
\caption{RMSE of SCH, EKF, KER, and UKF on the logistic equation using IWP($q$) priors for $q=1,\ldots,4$ plotted against step size.}\label{fig:experiment_logistic1:rmse}
\end{figure}

\begin{figure}[t!]
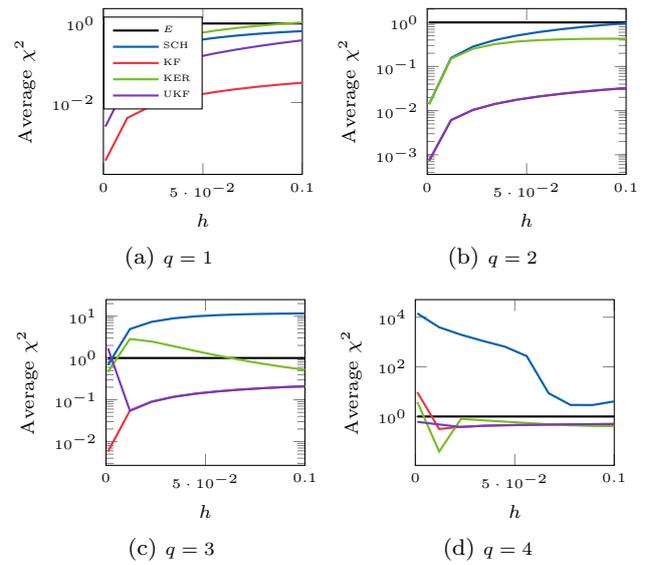

\centering
\subfloat[\scriptsize{$q = 1$}\label{fig:experiment_logistic1_chi2_q1}]{\input{fig/experiment_logistic1/experiment_logistic1_chi2_q1.tex}}
\subfloat[\scriptsize{$q = 2$}\label{fig:experiment_logistic1_chi2_q2}]{\input{fig/experiment_logistic1/experiment_logistic1_chi2_q2.tex}}\\
\subfloat[\scriptsize{$q = 3$}\label{fig:experiment_logistic1_chi2_q3}]{\input{fig/experiment_logistic1/experiment_logistic1_chi2_q3.tex}}
\subfloat[\scriptsize{$q = 4$}\label{fig:experiment_logistic1_chi2_q4}]{\input{fig/experiment_logistic1/experiment_logistic1_chi2_q4.tex}}
\caption{Average $\chi^2$-statistic of SCH, EKF, KER, and UKF on the logistic equation using IWP($q$) priors for $q=1,\ldots,4$ plotted against step size. The expected $\chi^2$-statistic is shown in black (E).}\label{fig:experiment_logistic1:chi2}
\end{figure}

\begin{figure}[t!]
\centering
\input{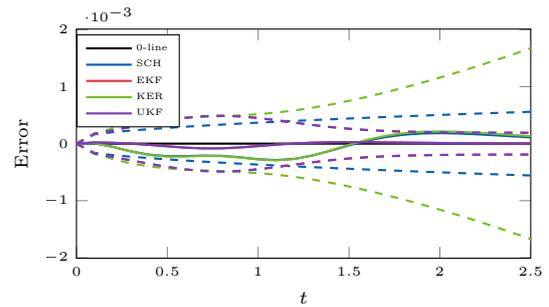}
\caption{The errors (solid lines) and $\pm$ 2 standard deviation bands (dashed) for KF, SCH, and KER on the logistic with $q=2$ and $h = 10^{-1}$. A line at 0 is plotted in solid black.}\label{fig:demonstration_logistic1}
\end{figure}

\subsection{The FitzHugh–-Nagumo Model}
The FitzHugh--Nagumo model is given by:
\begin{equation}
\begin{bmatrix} \dot{y}_1(t) \\ \dot{y}_2(t) \end{bmatrix} = \begin{bmatrix} c\left(y_1(t) - \frac{y_1^3(t)}{3} + y_2(t) \right)\\ - \frac{1}{c}\left(y_1(t) - a + by_2(t) \right) \end{bmatrix},
\end{equation}
where we set $(a,b,c) = (.2,.2,3)$ and  $y(0) = [-1\ 1]^\T$. As previous experiments showed that the behaviour of KER and UKF are similar to SCH and EKF, respectively, we opt for only comparing the latter to increase readability of the presented results. As previously, the moments of $X^{(1)}(0)$, $X^{(2)}(0)$, and $X^{(3)}(0)$ are initialised to their exact values and the remaining derivatives are initialised with zero mean and unit variance. The integration interval is set to $[0,20]$ and all methods use an IWP($q$) prior for $q=1,\ldots,4$ and  the uncertainty is calibrated as explained in Section \ref{susbec:non_affine_calibration}. A baseline solution is computed using MATLAB's \texttt{ode45} function with an absolute tolerance of $10^{-15}$ and relative tolerance of $10^{-12}$, all errors are computed under the assumption that \texttt{ode45} provides the exact solution. The methods are examined for 10 step sizes uniformly placed on the interval $[10^{-3},10^{-1}]$.

The RMSE is shown in Figure \ref{fig:experiment_fitzhugh1:rmse}. For $q=1$ EKF produces an error orders of magnitude larger than SCH and for $q = 2$ both methods produce similar errors until the step size grows too large, causing SCH to start producing orders of magnitude larger errors than EKF. For $q=3,4$ EKF is superior in producing lower errors and additionally SCH can be seen to become unstable for larger step-sizes (at $h \approx 5\cdot 10^{-2}$ for $q=3$ and at $h \approx 2\cdot 10^{-2}$ for $q=4$). Furthermore, the averaged $\chi^2$-statistic is shown in Figure \ref{fig:experiment_fitzhugh1:chi2}. It can be seen that EKF is overconfident for $q=1$ while SCH is underconfident. For $q=2$ both methods are underconfident while EKF remains underconfident for $q=3,4$ but SCH becomes overconfident for almost all step sizes.

The error trajectory for the first component of $y$ and the reported uncertainty of the solvers is shown in Figure \ref{fig:demonstration_fitzhugh1} for $h = 5\cdot10^{-2}$ and $q = 2$. It can be seen that both methods have periodically occurring spikes in their errors with EKF being larger in magnitude but also briefer. However, the uncertainty estimate of the EKF is also spiking at the same time giving an adequate assessments of its error. On the other hand, the uncertainty estimate of SCH grows slowly and monotonically over the integration interval, with the error estimate going outside the two standard deviation region at the first spike (slightly hard to see in the figure).

\begin{figure}[t!]
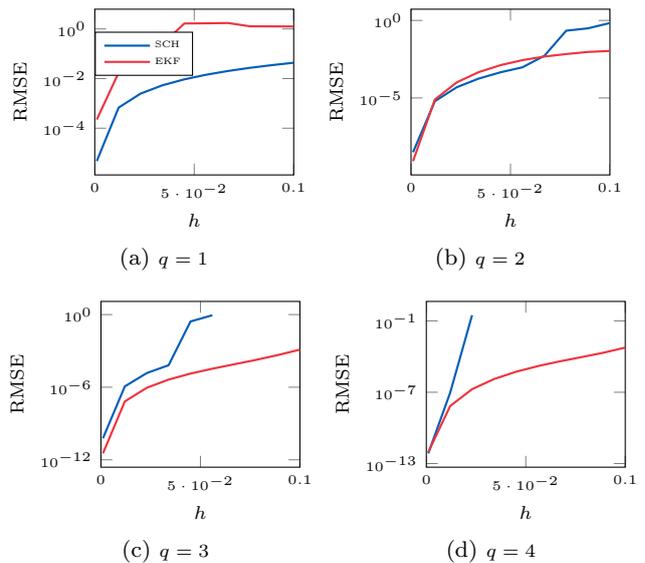

\centering
\subfloat[\scriptsize{$q = 1$}\label{fig:experiment_fitzhugh1_rmse_q1}]{\input{fig/experiment_fitzhugh1/experiment_fitzhugh1_rmse_q1.tex}}
\subfloat[\scriptsize{$q = 2$}\label{fig:experiment_fitzhugh1_rmse_q2}]{\input{fig/experiment_fitzhugh1/experiment_fitzhugh1_rmse_q2.tex}}\\
\subfloat[\scriptsize{$q = 3$}\label{fig:experiment_fitzhugh1_rmse_q3}]{\input{fig/experiment_fitzhugh1/experiment_fitzhugh1_rmse_q3.tex}}
\subfloat[\scriptsize{$q = 4$}\label{fig:experiment_fitzhugh1_rmse_q4}]{\input{fig/experiment_fitzhugh1/experiment_fitzhugh1_rmse_q4.tex}}
\caption{RMSE of SCH and EKF on the FitzHugh--Nagumo model using IWP($q$) priors for $q=1,\ldots,4$ plotted against step size.}\label{fig:experiment_fitzhugh1:rmse}
\end{figure}

\begin{figure}[t!]
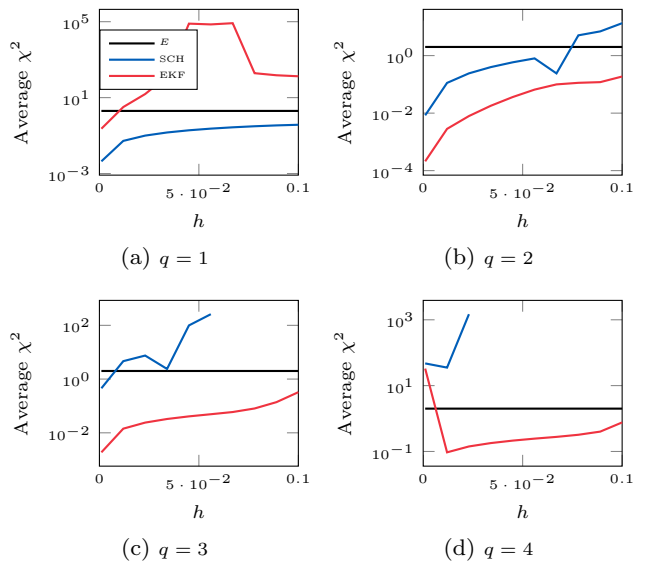

\centering
\subfloat[\scriptsize{$q = 1$}\label{fig:experiment_fitzhugh1_chi2_q1}]{\input{fig/experiment_fitzhugh1/experiment_fitzhugh1_chi2_q1.tex}}
\subfloat[\scriptsize{$q = 2$}\label{fig:experiment_fitzhugh1_chi2_q2}]{\input{fig/experiment_fitzhugh1/experiment_fitzhugh1_chi2_q2.tex}}\\
\subfloat[\scriptsize{$q = 3$}\label{fig:experiment_fitzhugh1_chi2_q3}]{\input{fig/experiment_fitzhugh1/experiment_fitzhugh1_chi2_q3.tex}}
\subfloat[\scriptsize{$q = 4$}\label{fig:experiment_fitzhugh1_chi2_q4}]{\input{fig/experiment_fitzhugh1/experiment_fitzhugh1_chi2_q4.tex}}
\caption{Average $\chi^2$-statistic of SCH and EKF on the FitzHugh--Nagumo model using IWP($q$) priors for $q=1,\ldots,4$ plotted against step size.}\label{fig:experiment_fitzhugh1:chi2}
\end{figure}

\begin{figure}[t!]
\centering
\input{fig/demonstration_fitzhugh1/demonstration_fitzhugh1_q2.tex}
\caption{The errors (solid lines) and $\pm$ 2 standard deviation bands (dashed) for KF, SCH, and KER on the FitzHugh--Nagumo model with $q=2$ and $h = 5\cdot10^{-2}$. A line at 0 is plotted in solid black.}\label{fig:demonstration_fitzhugh1}
\end{figure}

\subsection{A Bernoulli Equation}
In this following experiment we consider a transformation of Eq. \eqref{eq:experiment:logistic}, $\eta(t) = \sqrt{y(t)}$, for $r = 2$.
The resulting ODE for $\eta(t)$ now has two stable equilibrium points $\eta(t) = \pm 1$ and an unstable equilibrium point at $\eta(t) = 0$.
This makes it a simple test domain for different sampling-based ODE solvers, because different types of posteriors ought to arise.
We compare the proposed particle filter using both the proposal Eq. \eqref{eq:schober_proposal} (PF(1)) and EKF proposals
(Eq. \eqref{eq:ekf_proposal}) (PF(2)) with the method by \citep{o.13:_bayes_uncer_quant_differ_equat} (CHK) and the one by
\citep{conrad_probability_2017} (CON) for estimating $\eta(t)$ on the interval $t\in [0,5]$ with initial condition
set to $\eta_0 = 0$. Both PF and CHK use and IWP($q$) prior and set $R = \kappa h^{2q+1}$.
CON uses a Runge--Kutta method of order $q$ with perturbation variance $h^{2q+1}/[2q(q!)^2]$ as to
roughly match the incremental variance of the noise entering PF(1), PF(2), and CHK, which is determined by $Q(h)$ and not $R$.

First we attempt to estimate $y(5)=0$ for 10 step sizes uniformly placed on the interval $[10^{-3},10^{-1}]$ with $\kappa = 1$ and $\kappa = 10^{-10}$.
All methods use 1000 samples/particles and they estimate $y(5)$ by taking the mean over samples/empirical measures.
The estimate of $y(5)$ is plotted against the step size in Figure \ref{fig:experiment_bernoulli_mean1}. In general, the error increases with the step size
for all methods, though most easily discerned in Figures \ref{fig:experiment_bernoulli_mean1_q2_R0} and \ref{fig:experiment_bernoulli_mean1_q2_R1}.
All in all it appears that CHK, PF(1), and PF(2) behave similarly with regards to the estimation, while CON appears to produce a bit larger errors.
Furthermore, the effect of $\kappa$ appears to be the greatest on PF(1) and PF(2) as best illustrated in Figure \ref{fig:experiment_bernoulli_mean1_q1_R1}.

Additionally, kernel density estimates for the different methods are made for time points $t=1,3,5$ for $\kappa=1$,
$q=1,2$ and $h=10^{-1},5\cdot 10^{-2}$. In Figure \ref{fig:experiment_bernoulli_kde_h0_R0} kernel density estimates for $h = 10^{-1}$ are shown.
At $t = 1$ all methods produce fairly concentrated unimodal densities that then disperse as time goes on, with CON being a least concentrated
and dispersing quicker followed by PF(1)/PF(2) and then last CHK. Furthermore, CON goes bimodal as time goes on,
which is best seen in for $q=1$ in Figure \ref{fig:experiment_bernoulli_kde_t5_q1_h0_R0}.
On the other hand, the alternatives vary between unimodal (CHK in \ref{fig:experiment_bernoulli_kde_t5_q2_h0_R0}, also to some degree
PF(1) and PF(2)), bimodal (PF(1) and CHK in Figure \ref{fig:experiment_bernoulli_kde_t5_q1_h0_R0}), and even mildly trimodal
(PF(2) in Figure \ref{fig:experiment_bernoulli_kde_t5_q1_h0_R0}).

Similar behaviour of the methods is observed for $h = 5 \cdot 10^{-2}$ in Figure \ref{fig:experiment_bernoulli_kde_h0_R0},
though here all methods are generally more concentrated.

\begin{figure}[t!]
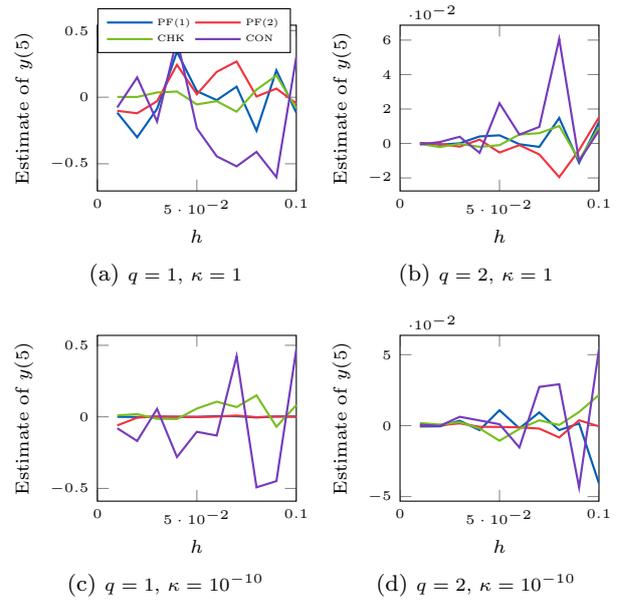

\centering
\subfloat[\scriptsize{$q=1$, $\kappa=1$}\label{fig:experiment_bernoulli_mean1_q1_R0}]{\input{fig/experiment_bernoulli_mean1/experiment_bernoulli_mean1_q1_R0.tex}}
\subfloat[\scriptsize{$q=2$, $\kappa=1$}\label{fig:experiment_bernoulli_mean1_q2_R0}]{\input{fig/experiment_bernoulli_mean1/experiment_bernoulli_mean1_q2_R0.tex}}\\
\subfloat[\scriptsize{$q=1$, $\kappa=10^{-10}$}\label{fig:experiment_bernoulli_mean1_q1_R1}]{\input{fig/experiment_bernoulli_mean1/experiment_bernoulli_mean1_q1_R1.tex}}
\subfloat[\scriptsize{$q=2$, $\kappa=10^{-10}$}\label{fig:experiment_bernoulli_mean1_q2_R1}]{\input{fig/experiment_bernoulli_mean1/experiment_bernoulli_mean1_q2_R1.tex}}
\caption{Sample mean estimate of the solution at $T = 5$. }\label{fig:experiment_bernoulli_mean1}
\end{figure}

\begin{figure}[t!]
\centering
\subfloat[\scriptsize{$q=1$, $t=1$}\label{fig:experiment_bernoulli_kde_t1_q1_h0_R0}]{\input{fig/experiment_bernoulli_kde_h0_R0/experiment_bernoulli_kde1_t1_q1_h0_R0.tex}}
\subfloat[\scriptsize{$q=2$, $t=1$}\label{fig:experiment_bernoulli_kde_t1_q2_h0_R0}]{\input{fig/experiment_bernoulli_kde_h0_R0/experiment_bernoulli_kde1_t1_q2_h0_R0.tex}}\\
\subfloat[\scriptsize{$q=1$, $t=3$}\label{fig:experiment_bernoulli_kde_t3_q1_h0_R0}]{\input{fig/experiment_bernoulli_kde_h0_R0/experiment_bernoulli_kde1_t3_q1_h0_R0.tex}}
\subfloat[\scriptsize{$q=2$, $t=3$}\label{fig:experiment_bernoulli_kde_t3_q2_h0_R0}]{\input{fig/experiment_bernoulli_kde_h0_R0/experiment_bernoulli_kde1_t3_q2_h0_R0.tex}}\\
\subfloat[\scriptsize{$q=1$, $t=5$}\label{fig:experiment_bernoulli_kde_t5_q1_h0_R0}]{\input{fig/experiment_bernoulli_kde_h0_R0/experiment_bernoulli_kde1_t5_q1_h0_R0.tex}}
\subfloat[\scriptsize{$q=2$, $t=5$}\label{fig:experiment_bernoulli_kde_t5_q2_h0_R0}]{\input{fig/experiment_bernoulli_kde_h0_R0/experiment_bernoulli_kde1_t5_q2_h0_R0.tex}}
\caption{Kernel density estimates of the solution of the Bernoulli equation for $h = 10^{-1}$ and $\kappa = 1$. Mind the different scale of the axes.}
\label{fig:experiment_bernoulli_kde_h0_R0}
\end{figure}

\begin{figure}[t!]
\centering
\subfloat[\scriptsize{$q=1$, $t=1$}\label{fig:experiment_bernoulli_kde_t1_q1_h1_R0}]{\input{fig/experiment_bernoulli_kde_h1_R0/experiment_bernoulli_kde1_t1_q1_h1_R0.tex}}
\subfloat[\scriptsize{$q=2$, $t=1$}\label{fig:experiment_bernoulli_kde_t1_q2_h1_R0}]{\input{fig/experiment_bernoulli_kde_h1_R0/experiment_bernoulli_kde1_t1_q2_h1_R0.tex}}\\
\subfloat[\scriptsize{$q=1$, $t=3$}\label{fig:experiment_bernoulli_kde_t3_q1_h1_R0}]{\input{fig/experiment_bernoulli_kde_h1_R0/experiment_bernoulli_kde1_t3_q1_h1_R0.tex}}
\subfloat[\scriptsize{$q=2$, $t=3$}\label{fig:experiment_bernoulli_kde_t3_q2_h1_R0}]{\input{fig/experiment_bernoulli_kde_h1_R0/experiment_bernoulli_kde1_t3_q2_h1_R0.tex}}\\
\subfloat[\scriptsize{$q=1$, $t=5$}\label{fig:experiment_bernoulli_kde_t5_q1_h1_R0}]{\input{fig/experiment_bernoulli_kde_h1_R0/experiment_bernoulli_kde1_t5_q1_h1_R0.tex}}
\subfloat[\scriptsize{$q=2$, $t=5$}\label{fig:experiment_bernoulli_kde_t5_q2_h1_R0}]{\input{fig/experiment_bernoulli_kde_h1_R0/experiment_bernoulli_kde1_t5_q2_h1_R0.tex}}
\caption{Kernel density estimates of the solution of the Bernoulli equation for $h = 5\cdot10^{-2}$ and $\kappa = 1$. Mind the different scale of the axes.}
\label{fig:experiment_bernoulli_kde_h1_R0}
\end{figure}

\section{Conclusion and Discussion}\label{sec:conclusion}
In this paper, we have presented a novel formulation of probabilistic numerical solution of ODEs as a standard problem in GP regression with a non-linear measurement function, and with measurements that are identically zero. The new model formulation enables the use of standard methods in signal processing to derive new solvers, such as EKF, UKF, and PF. We can also recover many of the previously proposed sequential probabilistic ODE solvers as special cases. 

Additionally, we have demonstrated excellent stability properties of the EKF and UKF on linear test equations, that is, A-stability has been established. The notion of A-stability is closely connected with the solution of stiff equations, which is typically achieved with \emph{implicit} or \emph{semi-implicit} methods \citep{Hairer1996}. In this respect our methods (EKF and UKF) most closely fit into the class of semi-implicit methods such as the methods of Rosenbrock type \citep[Chapter IV.7]{Hairer1996}. Though it does seem feasible the proposed methods can be nudged towards the class of implicit methods by means of iterative Gaussian filtering \citep{Bell1993,Garcia2015,Tronarp2018a}. 

While the notion of A-stability has been fairly successful in discerning between methods with good and bad stability properties, it is not the whole story \citep[Section 3]{Alexander1977}. This has lead to other notions of stability such as \emph{L-stability} and \emph{B-stability} \citep[Chapter IV.3 and IV.12]{Hairer1996}. It is certainly an interesting question whether the present framework allows for the development of methods satisfying these more strict notions of stability.

An advantage of our model formulation is the decoupling of the prior from the likelihood. Thus future work would involve investigating how well the exact posterior to our inference problem approximates the ODE and then analysing how well different approximate inference strategies behave. However, for $h\to 0$, we expect that the novel Gaussian filters (EKF,UKF) will exhibit polynomial worst-case convergence rates of the mean and its credible intervals, that is, its Bayesian uncertainty estimates, as has already been proved in \citep{Kersting2018} for 0-th order Taylor-series filters with arbitrary constant measurement variance $R$ (see Section \ref{sec:Taylor-Series_Methods}).

Our Bayesian recast of ODE solvers might also pave the way toward an average-case analysis of these methods, which has already been executed in \citep{Ritter2000} for the special case of Bayesian quadrature.
For the PF, a thorough convergence analysis similar to \citet{o.13:_bayes_uncer_quant_differ_equat}, \citet{conrad_probability_2017}, \citet{AbdulleGaregnani17} and \citet{del2004feynman} appears feasible.
However, the results on spline approximations for ODEs \citep[see, e.g.,][]{Loscalzo67} might also apply to the present methodology via the correspondence between GP regression and spline function approximations \citep{KimeldorfWahba70}.

\begin{acknowledgements}
This material was developed, in part, at the \textit{Prob Num 2018} workshop hosted by the Lloyd's Register Foundation programme on Data-Centric Engineering at the Alan Turing Institute, UK, and supported by the National Science Foundation, USA, under Grant DMS-1127914 to the Statistical and Applied Mathematical Sciences Institute.
Any opinions, findings, conclusions or recommendations expressed in this material are those of the author(s) and do not necessarily reflect the views of the above-named funding bodies and research institutions.
Filip Tronarp gratefully acknowledge financial support by Aalto ELEC Doctoral School. Additionally, Filip Tronarp and Simo S\"arkk\"a gratefully acknowledge financial support by Academy of Finland grant \#313708.  
Hans Kersting and Philipp Hennig gratefully acknowledge financial support by the German Federal Ministry of Education and Research through BMBF grant 01IS18052B (ADIMEM). 
Philipp Hennig also gratefully acknowledges support through ERC StG Action 757275 / PANAMA. Finally, the authors would like to thank the editor and the reviewers for their help in improving the quality of this manuscript.
\end{acknowledgements}

\appendix

\section{Proof of Proposition \ref{prop:BQ}}\label{appendix:proof1}
In this section we prove Proposition \ref{prop:BQ}. First note that, by Eq. \eqref{eq:model_assumption1}, we have
\begin{equation}\label{eq:kernel_means}
 \frac{\dif \C\big[X^{(1)}(t),X^{(2)}(s)\big]}{\dif t} = \C\big[X^{(2)}(t),X^{(2)}(s)\big],
\end{equation}
where $\C$ is the cross-covariance operator. That is the cross-covariance matrix between $X^{(1)}(t)$ and $X^{(2)}(t)$ is just the integral of the covariance matrix function of $X^{(2)}$.  Now define
\begin{subequations}
\begin{align}
(\mathbf{X}^{(i)})^\T &= \begin{bmatrix} \big(X^{(i)}_1\big)^\T & \dots & \big(X^{(i)}_N\big)^\T \end{bmatrix}, \ i=1,\dots,q+1,\\
\mathbf{g}^\T &= \begin{bmatrix} g^\T(h) & \dots & g^\T(Nh) \end{bmatrix},\\
\mathbf{z}^\T &= \begin{bmatrix} z_1^\T & \dots & z_N^\T \end{bmatrix}.
\end{align}
\end{subequations}
Since Equation \eqref{eq:ode_prior} defines a Gaussian process we have that $\mathbf{X}^{(1)}$ and $\mathbf{X}^{(2)}$ are jointly Gaussian distributed and from Eq. \eqref{eq:kernel_means} the blocks of $\C[\mathbf{X}^{(1)},\mathbf{X}^{(2)}]$ are given by
\begin{equation*}
\C\big[\mathbf{X}^{(1)},\mathbf{X}^{(2)}\big]_{n,m} =  \int_0^{nh} \C\big[X^{(2)}(t),X^{(2)}(mh)\big] \dif t
\end{equation*}
which is precisely the kernel mean, with respect to the Lebesgue measure on $[0,nh]$, evaluated at $mh$, see \cite[Section 2.2]{Briol2015probint}.
Furthermore,
\begin{equation*}
\V\big[\mathbf{X}^{(2)}\big]_{n,m} = \C\big[X^{(2)}(nh),X^{(2)}(mh)\big],
\end{equation*}
 that is, the covariance matrix function (referred to as kernel matrix in Bayesian quadrature literature \citep{Briol2015probint}) evaluated at all pairs in $\{h,\dots,Nh\}$. From Gaussian conditioning rules we have for the conditional means and covariance matrices given $\mathbf{X}^{(2)} - \mathbf{g} = 0$, denoted by $\E_{\mathcal{D}}[X^{(1)}(nh)]$ and $\V_{\mathcal{D}}[X^{(1)}(nh)]$, respectively, that
\begin{subequations}
\begin{align*}
\begin{split}
\E_{\mathcal{D}}\big[X^{(1)}(nh)\big] &= \E\big[X^{(1)}(nh)\big] + \mathbf{w}_n\big(\mathbf{z} + \mathbf{g} - \E\big[\mathbf{X}^{(2)}\big]\big)\\
 &= \E\big[X^{(1)}(nh)\big] + \mathbf{w}_n\big(\mathbf{g} - \E\big[\mathbf{X}^{(2)}\big]\big),
\end{split}\\
\V_{\mathcal{D}}\big[X^{(1)}(nh)\big] &= \V\big[X^{(1)}(nh)\big] - \mathbf{w}_n\V\big[\mathbf{X}^{(2)}\big]\mathbf{w}_n^\T,
\end{align*}
\end{subequations}
where we used the fact that $\mathbf{z} = 0$ by definition and $\mathbf{w}_n$ are the Bayesian quadrature weights associated to the integral of $g$ over the domain $[0,nh]$, given by (see \citealt[Proposition 1]{Briol2015probint})
\begin{equation*}
\mathbf{w}_n^\T = \V\big[\mathbf{X}^{(2)}\big]^{-1} \begin{bmatrix} \C\big[X^{(1)}(nh),X^{(2)}(h)\big]^\T \\ \vdots \\ \C\big[X^{(1)}(nh),X^{(2)}(Nh)\big]^\T  \end{bmatrix}.
\end{equation*}
\qed
\section{Proof of Proposition \ref{prop:kersting}}\label{appendix:proof3}
To prove Proposition \ref{prop:kersting}, expand the expressions for $S_n$ and $K_n$ as given by Eq. \eqref{eq:general_gaussian_update}:
\begin{subequations}
\begin{align*}
\begin{split}
S_n &= \dot{C} \Sigma_n^P \dot{C}^\T + \V\big[f(CX_n,t_n)\mid z_{1:n-1}\big] \\
&\quad - \dot{C}\C\big[X_n,f(CX_n,t_n)\mid z_{1:n-1}\big] \\
&\quad - \C\big[X_n,f(CX_n,t_n)\mid z_{1:n-1}\big]^\T\dot{C}^\T \\
&\approx \dot{C} \Sigma_n^P \dot{C}^\T + \V\big[f(CX_n,t_n)\mid z_{1:n-1}\big]
\end{split} \\
\begin{split}
K_n &= \big(\Sigma_n^P \dot{C}^\T - \C\big[X_n,f(CX_n,t_n)\mid z_{1:n-1}\big]\big)S_n^{-1} \\
&\approx \Sigma_n^P \dot{C}^\T S_n^{-1},
\end{split}
\end{align*}
\end{subequations}
where in the second steps the approximation $\C[X_n,f(CX_n,t_n)\mid z_{1:n-1}] \approx 0$ was used.
Lastly, recall that $z_n \triangleq 0$, hence the update equations become
\begin{subequations}
\begin{align}
S_n &\approx \dot{C} \Sigma_n^P \dot{C}^\T + \V\big[f(CX_n,t_n)\mid z_{1:n-1}\big], \\
K_n &\approx  \Sigma_n^P \dot{C}^\T S_n^{-1},\\
\mu_n^F &\approx \mu_n^P + K_n\big( \E\big[f(CX_n,t_n)\mid z_{1:n-1}\big] - \dot{C}\mu_n^P\big), \\
\Sigma_n^F &\approx \Sigma_n^P - K_n S_n K_n^\T.
\end{align}
\end{subequations}
When $\E[f(CX_n,t_n)\mid z_{1:n-1}]$ and $\V[f(CX_n,t_n)\mid z_{1:n-1}]$ are approximated by Bayesian quadrature using a squared exponential kernel and a uniform set of nodes translated and scaled by $\mu_n^P$ and $\Sigma_n^P$, respectively, the method of \citet{Kersting2016UAI} is obtained.\qed

\section{Proof of Proposition \ref{prop:s2mle_exact}}\label{appendix:proof4}
Note that $(\breve{\mu}_n^F,\breve{\Sigma}_n^F)$ is the output of a misspecified Kalman filter \cite[Algorithm 1]{Tronarp2019a}. We indicate that a quantity from Eqs. \eqref{eq:kalman_prediction} and \eqref{eq:general_gaussian_update} is computed by the misspecified Kalman filter by $\breve{}$. For example $\breve{\mu}_n^P$ is the predictive mean of the misspecified Kalman filter. If $\Sigma_n^F = \sigma^2 \breve{\Sigma}_n^F$ and $\breve{\mu}_n^F = \mu_n^F$ 
holds then for the prediction step we have
\begin{subequations}
\begin{align*}
\mu_{n+1}^P &= A(h)\mu_n^F + \xi(h) = A(h)\breve{\mu}_n^F + \xi(h) = \breve{\mu}_{n+1}^P, \\
\begin{split}
\Sigma_{n+1}^P &= A(h)\Sigma_n^F A^\T(h)  + Q(h), \\
 &= \sigma^2 \Big( A(h)\breve{\Sigma}_n^F A^\T(h)  + \breve{Q}(h) \Big), \\
&= \sigma^2 \breve{\Sigma}_{n+1}^P, 
\end{split}
\end{align*}
\end{subequations} 
where we used the fact that $Q(h) = \sigma^2 \breve{Q}(h)$, which follows from $L = \sigma \breve{L}$ and Eq. \eqref{eq:lti_disc}. Furthermore, recall that $H_{n+1} = \dot{C} - \Lambda(t_{n+1}) C$,  which for the update gives 
\begin{subequations}
\begin{align*}
\begin{split}
S_{n+1} &= H_{n+1} \Sigma_{n+1}^P H_{n+1}^\T \\
&= \sigma^2 H_{n+1} \breve{\Sigma}_{n+1}^P H_{n+1}^\T \\
&= \sigma^2 \breve{S}_{n+1}.
\end{split} \\
\begin{split}
K_{n+1} &= \Sigma_{n+1}^P H_{n+1}^\T  S_{n+1}^{-1} \\
&= \sigma^2 \breve{\Sigma}_{n+1}^PH_{n+1}^\T [\sigma^2 \breve{S}_{n+1}]^{-1} \\
&= \breve{\Sigma}_{n+1}^PH_{n+1}^\T \breve{S}_{n+1}^{-1} \\
&= \breve{K}_{n+1}. 
\end{split} \\
\begin{split}
\hat{z}_{n+1} &=  H_{n+1} \mu_{n+1}^P - \zeta(t_n) \\
&= H_{n+1} \breve{\mu}_{n+1}^P  - \zeta(t_n) \\
&= \breve{z}_{n+1},
\end{split} \\
\begin{split}
\mu_{n+1}^F &= \mu_{n+1}^P + K_{n+1}(z_{n+1} - \hat{z}_{n+1}) \\
&= \breve{\mu}_{n+1}^P + \breve{K}_{n+1}(z_{n+1} - \breve{z}_{n+1}) \\
&= \breve{\mu}_{n+1}^F. 
\end{split}\\
\begin{split}
\Sigma_{n+1}^F &= \Sigma_{n+1}^P - K_{n+1} S_{n+1} K_{n+1}^\T \\
&= \sigma^2 \Big( \breve{\Sigma}_{n+1}^P - \breve{K}_{n+1} \breve{S}_{n+1} \breve{K}_{n+1}^\T \Big) \\
&= \sigma^2\breve{\Sigma}_{n+1}^F. 
\end{split}
\end{align*}
\end{subequations}
It thus follows by induction that $\mu_n^F = \breve{\mu}_n^F$, $\Sigma_n^F = \sigma^2 \breve{\Sigma}_n^F$, $\hat{z}_n = \breve{z}_n$, and $S_n = \sigma^2 \breve{S}_n$ for $n \geq 0 $. From Eq. \eqref{eq:marginal_likelihood} 
we have that the log-likelihood is given by  
\begin{equation*}
\begin{split}
\log p(z_{1:N}) &= \log \prod_{n=1}^N \mathcal{N}(z_n; \hat{z}_n,S_n) \\
 &= \log \prod_{n=1}^N \mathcal{N}(z_n; \breve{z}_n,\sigma^2 \breve{S}_n) \\
&= -\frac{Nd}{2}\log \sigma^2 \\
&\quad -\sum_{n=1}^N \frac{(z_n-\breve{z}_n)^\T \breve{S}_n^{-1} (z_n-\breve{z}_n)}{2\sigma^2}.
\end{split}
\end{equation*}
Taking the derivative of log-likelihood with respect to $\sigma^2$ and setting it to zero gives the following estimating equation
\begin{equation*}
0 = -\frac{Nd}{2\sigma^2} + \frac{1}{2(\sigma^2)^2}\sum_{n=1}^N(z_n-\breve{z}_n)^\T \breve{S}_n^{-1} (z_n-\breve{z}_n), 
\end{equation*}
which has the following solution
\begin{equation*}
\sigma^2 = \frac{1}{Nd} \sum_{n=1}^N (z_n - \breve{z}_n)^\T \breve{S}_n^{-1} (z_n - \breve{z}_n).
\end{equation*}
\qed

\bibliographystyle{spbasic}

\end{document}